\documentclass[a4paper,11pt]{article}

\usepackage[mathlines]{lineno}
\usepackage{amsmath}
\usepackage{amssymb}
\usepackage{amsthm}
\usepackage{graphicx}
\usepackage{enumerate}
\usepackage{cite}
\usepackage{fullpage}
\usepackage{datetime2}
\usepackage{todonotes}
\usepackage{hyperref}
\usepackage{scalerel}

\allowdisplaybreaks

\newcommand{\conv}{\ensuremath{\mathrm{conv}}}
\newcommand{\Real}{\ensuremath{\mathbb{R}}}
\newcommand{\Plane}{\ensuremath{\mathbb{R}^2}}

\newcommand{\dual}[1]{\ensuremath{#1^\star}}
\newcommand{\arr}{\ensuremath{\mathcal{A}}}

\newcommand{\LE}{\mathcal{L}}
\newcommand{\UE}{\mathcal{U}}

\newcommand{\width}{\mathrm{width}}
\newcommand{\strip}{\sigma}
\newcommand{\stripbar}{\overline{\strip}}
\newcommand{\ASDS}{\mathcal{D}}
\newcommand{\wparam}{\omega}
\newcommand{\CWDS}{\mathcal{W}}
\newcommand{\wf}{\omega}
\newcommand{\wfg}{\overline{\omega}}
\newcommand{\upd}{\upsilon}
\newcommand{\Upd}{\Upsilon}

\newcommand{\cL}{\mathcal{L}}

\newtheoremstyle{mytheorem}{3pt}{3pt}{\slshape}{}{\bfseries}{}{.5em}{}
\theoremstyle{mytheorem}
\newtheorem{lemma}{Lemma}
\newtheorem{theorem}{Theorem}

\theoremstyle{definition}

\makeatletter
\newbox\ProofSym
\setbox\ProofSym=\hbox{%
\unitlength=0.18ex%
\begin{picture}(10,10)
\put(0,0){\framebox(9,9){}} \put(0,3){\framebox(6,6){}}
\end{picture}}
\renewenvironment{proof}[1][Proof.]{\O@proof{#1}}{\O@endproof}
\def\O@proof#1{\trivlist
   \@topsep\z@\@topsepadd\smallskipamount%
   \@ifstar{\item[]}{\item[\hskip\labelsep\it #1 ]}}
\def\O@endproof{\hfill\copy\ProofSym\linebreak[3mm]\endtrivlist}
\makeatother

\let\geq\geqslant
\let\leq\leqslant

\def\denseitems{
    \itemsep1pt plus1pt minus1pt
    \parsep0pt plus0pt
    \parskip0pt\topsep0pt}

\newcommand*\patchAmsMathEnvironmentForLineno[1]{%
  \expandafter\let\csname old#1\expandafter\endcsname\csname #1\endcsname
  \expandafter\let\csname oldend#1\expandafter\endcsname\csname end#1\endcsname
  \renewenvironment{#1}%
     {\linenomath\csname old#1\endcsname}%
     {\csname oldend#1\endcsname\endlinenomath}}%
\newcommand*\patchBothAmsMathEnvironmentsForLineno[1]{%
  \patchAmsMathEnvironmentForLineno{#1}%
  \patchAmsMathEnvironmentForLineno{#1*}}%
\AtBeginDocument{%
\patchBothAmsMathEnvironmentsForLineno{equation}%
\patchBothAmsMathEnvironmentsForLineno{align}%
\patchBothAmsMathEnvironmentsForLineno{flalign}%
\patchBothAmsMathEnvironmentsForLineno{alignat}%
\patchBothAmsMathEnvironmentsForLineno{gather}%
\patchBothAmsMathEnvironmentsForLineno{multline}%
}

\title{Constrained Two-Line Center Problems
\thanks{T. Ahn was supported by the Institute of Information \& communications Technology Planning \& Evaluation(IITP) grant funded by the Korea government(MSIT) (No. 2019-0-01906).
S.W.Bae was supported by the National Research Foundation of Korea(NRF) grant funded by the Korea government(MSIT) (No. RS-2023-00251168).}} 

\author{Taehoon Ahn\thanks{Graduate School of Artificial Intelligence, Pohang University of Science and Technology, Pohang, Republic of Korea. {\tt sloth@postech.ac.kr}}
\and
{Sang Won} Bae\thanks{Division of Artificial Intelligence and Computer Science, Kyonggi University, Suwon, Korea. {\tt swbae@kgu.ac.kr}}}

\begin{document}
\date{\DTMnow}
\maketitle

\begin{abstract}
 Given a set~$P$ of $n$~points in the plane,
 the two-line center problem asks to find two lines that minimize
 the maximum distance from each point in~$P$ to its closer one of the two resulting lines.
 The currently best algorithm for the problem takes $O(n^2 \log^2 n)$ time
 by Jaromczyk and Kowaluk in 1995.
 In this paper, 
 we present faster algorithms for three variants of the two-line center problem 
 in which the orientations of the resulting lines are constrained.
 Specifically,
 our algorithms solve the problem in $O(n \log n)$ time when the orientations of both lines are fixed;
 in $O(n \log^3 n)$ time when the orientation of one line is fixed;
 and in $O(n^2 \alpha(n) \log n)$ time when the angle between the two lines is fixed,
 where $\alpha(n)$ denotes the inverse Ackermann function.
\end{abstract}

\section{Introduction} \label{sec:intro}

Given a set~$P$ of $n$~points in the plane~$\Plane$,
the \emph{two-line center problem} asks to find two lines that minimize
the maximum distance from each point in~$P$ to its closer one of the two resulting lines.
In~1991, Agarwal and Sharir~\cite{as-pglpmwps-91} presented 
the first subcubic $O(n^2 \log^5 n)$-time algorithm for the two-line center problem,
in which they solved the decision version in $O(n^2 \log^3 n)$ time
using their machinery~\cite{as-odmwpps-91}
to maintain the width of a point set under a prescribed sequence of changes
and then to apply the parametric search technique. 
(See also its full version~\cite{as-pglp-94}.)
In 1995, Jaromczyk and Kowaluk~\cite{jk-tlcppv:nads-95} presented an $O(n^2 \log^2 n)$-time algorithm
and also discussed an $O(n^2 \log n)$-time decision algorithm.
Glozman et al.~\cite{gks-sgsopsm-95, gks-sgsopsm-98} exhibited 
how any $D$-time decision algorithm for the two-line center problem
can be converted to an optimization algorithm of $O(n^2 \log n + D\log n)$ time
using sorted matrices.
Later, Katz and Sharir~\cite{ks-ebago-97} introduced
an expander-based approach and showed how to solve the problem in
$O(n^2 \log^3 n + D \log n)$ time.
There was no significant progress since then
and $O(n^2 \log^2 n)$ still remains the best known upper bound~\cite{jk-tlcppv:nads-95,gks-sgsopsm-98}.

This paper addresses constrained variants of the two-line center problem,
and aims to provide efficient algorithms for the constrained problems, 
particularly faster than $O(n^2 \log^2 n)$ time,
and to provide new observations and algorithmic techniques for any future breakthrough on the problem.
The currently fastest algorithm by Jaromczyk and Kowaluk~\cite{jk-tlcppv:nads-95} indeed considers
several constrained problems, tackled by different methods.
Though not having explicitly mentioned in~\cite{jk-tlcppv:nads-95}, their approach yields
an $O(n \log^2 n)$-time algorithm when a fixed point in~$P$ should be the farthest
to the resulting lines, after an $O(n^2)$-time preprocessing.
Recently, Bae~\cite{Bae2020} presented an $O(n^2)$-time algorithm
for the two-\emph{parallel}-line center problem, in which the two resulting lines
are supposed to be parallel.

In this paper, we solve three variants of the two-line center problem,
constrained about the orientations of the resulting two lines.
Following summarizes our results and approaches:
\begin{enumerate}[(1)]
\denseitems
\item \textit{(Two fixed orientations)} 
 Given two orientations~$\theta$ and $\phi$,
 we present an $O(n \log n)$-time algorithm that solves
 the two-line center problem in which the two resulting lines are constrained to
 have orientations~$\theta$ and~$\phi$.
 If the input points~$P$ are given as a sorted list
 in one of the specified orientations, then the running time can be reduced to $O(n)$.
\item \textit{(One fixed orientation)}
 Given an orientation~$\phi$,
 we present an $O(n \log^3 n)$-time algorithm that solves
 the two-line center problem in which one of the resulting lines is constrained to
 have orientation~$\phi$.
 We first devise an $O(n\log^2 n)$-time decision algorithm for this constrained problem
 using the data structure by Agarwal and Sharir~\cite{as-odmwpps-91}.
 In spite of having such an efficient decision algorithm,
 it is not immediate to achieve a sub-quadratic time optimization algorithm
 by applying known techniques; as introduced above,
 all known techniques for the two-line center problem require at least quadratic-time 
 additional overhead~\cite{gks-sgsopsm-98,as-pglp-94,ks-ebago-97}. 
 To overcome this difficulty,
 we use our decision algorithm as a subroutine to find an interval narrow enough to reduce
 the possible number of candidate configurations to~$O(n)$ and 
 apply the dynamic width structure by Chan~\cite{c-fdapw-03}.
\item \textit{(Fixed angle of intersection)}
 Given a real~$\beta$,
 we present an $O(n^2 \alpha(n) \log n)$-time algorithm that solves
 the two-line center problem in which the two resulting lines are constrained to
 make angle~$\beta$, where $\alpha(n)$ denotes the inverse Ackermann function.
 As in the second problem, we start by presenting a decision algorithm 
 and apply the known technique~\cite{gks-sgsopsm-98}
 to obtain a favorably narrow interval that contains the optimal width value.
 We then consider a sweeping process in which we rotate a strip of variable width within the interval,
 and prove that if suffices to find an optimal solution by simulating the process.
\end{enumerate}

To our best knowledge, the three constrained problems 
have not been considered in the literature. 
Note that the two-parallel-line center problem studied in~\cite{Bae2020} is a more constrained variant
of our problems:
In the first problem (of two fixed orientations),
the special case of $\theta = \phi$ can be solved in $O(n)$~time,
and the third problem (of fixed angle) for $\beta = 0$ indeed asks to find
a two-parallel-line center, which can be solved in $O(n^2)$~time~\cite{Bae2020}.

\paragraph*{Related work.}
The two-line center problem is a special case of the \emph{$k$-line center} problem for~$k\geq 1$.
For~$k=1$, known as the \emph{width} problem,
one can solve the problem in $O(n \log n)$ time~\cite{ps-cgi-85},
or in $O(n)$ time if the convex hull of~$P$ is given~\cite{t-sgprc-83}.
In three dimensions, the width of $n$~points in~$\Real^3$
can be computed in $O(n^{3/2+\epsilon})$ expected time by Agarwal and Sharir~\cite{as-erasgop-96}.
In higher dimensions~$d\geq 4$, Chan~\cite{Chan2002} showed how to compute the width
in~$O(n^{\lceil d/2 \rceil})$ time.
In the plane~$\Plane$, the $k$-line center problem is known to be NP-hard
when $k$ is part of the input~\cite{mt-cllfp-82},
while efficient approximation algorithms are known~\cite{apv-aa2lc-03,apv-aaklc-05}.
Agarwal et al.~\cite{apv-aaklc-05} presented an efficient approximation algorithm.
Exact algorithms for~$k\leq 2$ are presented as aforementioned,
while any nontrivial exact algorithm for~$k\geq 3$ is, however, unknown.
An efficient $(1+\epsilon)$-approximation algorithm
for~$k=2$ is presented by Agarwal et al.~\cite{apv-aa2lc-03}.
Very recently, several constrained variants of the $k$-line center problem
and its generalization in high dimensions have been considered.
Das et al.~\cite{ddm-aaolc-23} presented an approximation algorithm
for the $k$-line center problem where the resulting lines are constrained to be axis-parallel.
Chung et al.~\cite{caba-plcgs-23} considered a variant of the parallel $k$-line center problem.
Ahn et al.~\cite{acabcy-mwdsweshd-24} presented first algorithms for the problem of finding
two parallel \emph{slabs} in~$\Real^d$ for $d\geq 3$.

Not being restricted to the line center problems,
there have been an enormous amount of results on constrained variants of
those problems of finding optimal locations of one or more geometric shapes enclosing input objects.
Such results on constrained problems usually
provided more efficient solutions than those for the original (unconstrained) problems
or played important roles as stepping stones to later breakthroughs.
Constrained two-square problems~\cite{kks-cscp-98} and
the problem of covering points by two disjoint rectangles~\cite{kba-cpstdr-11} are such examples.
%

\paragraph{Preliminaries.}

A \emph{strip}~$\strip$ is the closed region between two parallel lines
and its \emph{width} is the distance between the two lines.
A pair of two strips will be called a \emph{two-strip}
and the width of a two-strip mean the larger width of its two members.
Note that the two-line center problem is equivalent to the problem of finding
a two-strip of minimum width that enclose given points.
The \emph{orientation} of a line is a real value~$\theta \in [0, \pi)$
such that $\theta$ is the angle swept from a horizontal line in counterclockwise direction to the line.
Similarly, a strip is said to have an orientation~$\theta$ 
when its bounding lines are in orientation~$\theta$.
For any set~$P$ of points and orientation~$\theta \in [0,\pi)$,
we denote by $\strip_\theta(P)$ the minimum-width strip in orientation~$\theta$ that encloses~$P$.
We denote $\width_\theta(P) := \width(\strip_\theta(P))$.
The \emph{width} of point set~$P$, denoted by~$\width(P)$, is 
the smallest width of a strip that encloses~$P$.

\section{Two fixed orientations} \label{sec:2fixed}

In this section, we consider the first constrained problem
where the orientations of two line centers should given values~$\theta$ and~$\phi$.
We assume that $\phi = 0$ without loss of generality.

We start by sorting the $n$~points in~$P$ in the nondecreasing order of $y$-coordinates
and let $p_1, \ldots, p_n \in P$ be in this order.
For $0\leq i \leq n$, 
let $P_i := \{p_1, \ldots, p_i\}$ and $\overline{P}_i := P\setminus P_i = \{p_{i+1}, \ldots, p_n\}$.
It is straightforward in $O(n)$~time to incrementally construct 
the strips~$\strip_\theta(P_i)$ and $\strip_\theta(\overline{P}_i)$ in orientation~$\theta$
for all $0\leq i\leq n$.
We then observe the following.
\begin{lemma} \label{lem:2fixed_query}
 Given $P$ as a sorted list as above,
 $P$ can be processed in $O(n)$-time so that
 $\strip_\theta(P_i \cup \overline{P}_j)$
 can be answered in $O(1)$~time
 for any query pair~$(i,j)$ of indices.
\end{lemma}
\begin{proof}
As a preprocessing, we compute $\strip_\theta(P_i)$ and $\strip_\theta(\overline{P}_i)$ 
for all $0\leq i\leq n$, and store them in an array.
This can be done in $O(n)$ time using $O(n)$ space.
For any query pair $0\leq i \leq j\leq n$ of indices,
the strip~$\strip_\theta(P_i \cup \overline{P}_j)$ enclosing~$P_i \cup \overline{P}_j$ 
is simply the one enclosing both $\strip_\theta(P_i)$ and $\strip_\theta(\overline{P}_i)$,
so can be reported in $O(1)$~time.
\end{proof}

Consider any minimum-width two-strip $(\strip_1, \strip_2)$ enclosing~$P$
such that its orientations are $0$ and $\theta$, respectively.
Observe that $\strip_1$ includes a contiguous sequence $p_{i+1}, \ldots, p_{j}$ of points in~$P$
for some indices $0\leq i\leq j\leq n$,
while $\strip_2$ covers the points in $P_i \cup \overline{P}_j$.
Hence, the problem can be solved by searching for $O(n^2)$ possible bipartitions of~$P$,
namely, $(P_i \cup \overline{P}_j, P\setminus (P_i\cup \overline{P}_j))$ for $0\leq i\leq j\leq n$,
and evaluating the widths of the two strips enclosing each part of desired bipartitions.

Let $w_1(i,j) := \width_0(\{p_{i+1}, \ldots, p_j\})$ 
be the width of the smallest horizontal strip enclosing
$j-i$~points $p_{i+1}, \ldots, p_j \in P$, that is, 
the difference of the $y$-coordinates of $p_{i+1}$ and $p_j$.
Let $w_2(i,j) := \width_\theta(P_i \cup \overline{P}_j)$
be the width of the smallest strip in orientation~$\theta$ enclosing
$P_i \cup \overline{P}_j = \{p_1, \ldots, p_i, p_{j+1}, \ldots, p_n\}$.
Define $w(i,j) := \max\{w_1(i,j), w_2(i,j)\}$.
Our task is to minimize $w(i,j)$ over all $0\leq i \leq j\leq n$.
This can be done by evaluating $w_1(i,j)$ and $w_2(i,j)$ for at most $4n$ pairs $(i,j)$ of indices
due to the monotonicity of $w_1$ and $w_2$.
More precisely, observe that
\[ w_1(i,j) \leq w_1(i, j+1) \quad \text{and} \quad w_1(i,j)\leq w_1(i-1, j),\]
while
\[ w_2(i,j) \geq w_2(i, j+1) \quad \text{and} \quad w_2(i,j) \geq w_2(i-1, j)\]
by definition.
Hence, our algorithm initially sets~$i=j=0$ and repeatedly increases~$j$ by one until
it holds that $w_1(0,j) \leq w_2(0,j)$ and $w_1(0, j+1) \geq w_2(0, j+1)$.
Then for each $i=1,\ldots, n$ in this order, it repeatedly increases~$j$ by one
until it holds that $w_1(i,j) \leq w_2(i,j)$ and $w_1(i, j+1) \geq w_2(i, j+1)$
for the current~$i$.
This way, our algorithm probes at most $4n$ pairs $(i, j)$.
For a given pair $(i, j)$, in $O(1)$ time
we can evaluate $w_1(i, j)$ by definition and $w_2(i,j)$ by Lemma~\ref{lem:2fixed_query}.
We thus conclude the following.

\begin{theorem} \label{thm:2fixed}
 Given a set~$P$ of $n$~points and two orientations~$\theta, \phi \in [0,\pi)$,
 the two-line center problem where the resulting lines have orientations~$\theta$ and $\phi$
 can be computed in $O(n \log n)$ time, or in $O(n)$ time,
 provided $P$ is sorted in orientation either~$\theta$ or~$\phi$.
\end{theorem}

\section{One fixed orientation} \label{sec:1fixed}

In this section, we solve the second constrained problem:
given a fixed orientation~$\phi$,
find two strips of minimum width whose union encloses~$P$
such that one of the two strips is in orientation~$\phi$.
Throughout this section, 
a pair of two such strips $(\strip_1, \strip_2)$, where $\strip_1$ is in orientation~$\phi$,
will be simply called a \emph{constrained two-strip},
and assume that~$\phi=0$.

To find a constrained two-strip of minimum width enclosing~$P$,
one could make use of a data structure for the dynamic width maintenance~\cite{e-idmpw-00,c-fdapw-03}.
Observe that there are $O(n^2)$ possible bipartitions of~$P$
induced by a constrained two-strip $(\strip_1, \strip_2)$
since there are $O(n^2)$ distinct subsets of~$P$ that can be enclosed by a horizontal strip~$\strip_1$.
This approach, however, does not seem to avoid a quadratic running time,
since the point set we would maintain undergoes $\Theta(n^2)$ updates.
Another common approach is to apply known techniques,
such as the parametric search~\cite{as-pglp-94}, the expander-based method~\cite{ks-ebago-97}, 
or the one based on a sorted matrix~\cite{gks-sgsopsm-98}.
These techniques also require at least quadratic time overhead.

Despite this difficulty, we present a near-linear $O(n\log^3 n)$-time algorithm.
We first consider the decision problem, for a real parameter~$\wparam > 0$,
asking if there exists a constrained two-strip of width~$\wparam$ enclosing~$P$.
Note that the decision problem can be solved in $O(n\log^3 n)$ time
by a direct application of the machinery of Agarwal and Sharir~\cite{as-odmwpps-91}.
In the following, we show how to shave another logarithmic factor,
while still using the data structure of Agarwal and Sharir.
For the purpose, we first give a brief review on their data structure,
and make some geometric observations, which will be essential 
for the efficiency and correctness of our algorithm.
Once we have an efficient decision algorithm,
we can obtain a sufficiently narrow interval of possible width values such that 
there are only $O(n)$ possible bipartitions of~$P$ we should test.
We then use Chan's structure of dynamic width maintenance~\cite{c-fdapw-03}
to find an optimal two-strip.

\subsection{Data structures for dynamic width decision and maintenance}

Agarwal and Sharir~\cite{as-odmwpps-91} showed that in $O(n \log^3 n)$ time
the \emph{offline dynamic width decision problem} can be solved:
Given a parameter~$\wparam > 0$ and a sequence of 
$n$~insert/delete operations on a set~$S$ of points,
initially consisting of at most $n$~points,
determine whether there is any moment such that $\width(S) \leq \wparam$
during the $n$~updates on~$S$.
Their algorithm builds a segment tree based on the life-spans of the points,
that is, the time intervals in which each point is a member of~$S$,
and traverse it with a secondary data structure~$\ASDS$ 
that maintains necessary information about the width of the current~$S$
using linear space.

The data structure~$\ASDS$ consists of two balanced binary search trees
that store the edges of the convex hull~$\conv(S)$, ordered by their orientations,
and maintains a certain collection of invariants,
which suffice to decide in $O(1)$ time whether or not $\width(S) \leq \wparam$
for the current set~$S$.
Agarwal and Sharir showed that how to update~$\ASDS$
per insertion of a point into~$S$,
and also how to undo the latest insertion, recovering the structure~$\ASDS$
to the status before the latest insertion.
Summarizing, we have:
\begin{lemma}[Agarwal and Sharir~\cite{as-odmwpps-91}] \label{lem:ASDS1}
 Suppose the data structure~$\ASDS$ with a parameter~$\wparam$ 
 has been built on a set~$S$ of $n$~points.
 Then, we can decide whether or not $\width(S) \leq \wparam$ in $O(1)$ time, and 
 $\ASDS$ can be maintained in $O(\log^2 n)$ worst-case time for the following updates:
 inserting a point to~$S$ and undoing the latest insertion.
\end{lemma}

From these facts,
it is not difficult to see that
the data structure~$\ASDS$ can be applied to solve
\emph{the LIFO(Last-In-First-Out) online dynamic width decision problem}:
Given a parameter~$\wparam>0$ and a set~$S$, initially being empty,
determine whether there is any moment such that $\width(S) \leq \wparam$
in the course of a sequence of LIFO online insertions/deletions on~$S$.
A sequence of online insertions/deletions on~$S$ is called \emph{LIFO}
if each insertion can be arbitrary 
while each deletion occurs always for the last inserted point in~$S$.
\begin{lemma} \label{lem:ASDS}
 The data structure~$\ASDS$ with a given parameter~$\wparam$ on a set~$S$, initially being empty,
 can be maintained in total $O(n\log^2 n)$ time and $O(n)$ space
 under $n$~LIFO online updates on~$S$.
 Thus, in the same time bound, 
 we can solve the LIFO online dynamic width decision problem.
\end{lemma}
\begin{proof}
Initialize the data structure~$\ASDS$ for an empty set~$S$ with parameter~$\wparam$.
For each online update on~$S$,
we perform the following:
If it is an insertion of a point~$p$ to~$S$,
then we update~$\ASDS$ to maintain its invariants
and also record the necessary information to recover the structure
as described in~\cite{as-odmwpps-91}.
If it is a deletion,
then it should be about the last inserted point~$p\in S$,
so we can undo the insertion of~$p$ using the information stored at the moment of its insertion.
By Lemma~\ref{lem:ASDS1}, the total time complexity is $O(n\log^2 n)$
and the space complexity is $O(n)$.

By maintaining the data structure~$\ASDS$,
deciding whether $\width(S) \leq \wparam$ for the current~$S$ 
can be done in $O(1)$ time.
So, the LIFO online dynamic width decision problem can be solved
in $O(n\log^2 n)$ time and $O(n)$ space.
\end{proof}

Chan~\cite{c-fdapw-03} presented how to exactly maintain $\width(S)$
over fully online updates on~$S$.
Its amortized time per update is $O(\sqrt{n} \log^3 n)$,
based on the following data structure.
\begin{lemma}[Chan~\cite{c-fdapw-03}] \label{lem:Chan_width}
 There is a data structure~$\CWDS$ for a set~$S$ of $n$~points
 that supports deletions of points from~$S$ and queries of the following kind:
 given a query point set~$Q$, report $\width(S\cup Q)$.
 The total preprocessing and deletion time is $O(n \log^3 n)$
 and the query time is $O(|Q|\log^3 (n+|Q|))$.
 The space required for maintaining the structure~$\CWDS$ is $O(n \log n)$.
\end{lemma}
Though Chan did not discuss the space requirement for his method,
it is not difficult to see it as stated above from construction~\cite{c-fdapw-03}.
Note that if the online updates are deletions only,
then this results in an $O(n \log^3 n)$-time algorithm
that maintains the exact width.

\subsection{Orientation-constrained width}
Let $S$ be a set of points,
and let $\theta_1 \leq \theta_2$ be two orientations.
Define the \emph{$[\theta_1, \theta_2]$-constrained width} of~$S$ to be
 \[ \width_{[\theta_1, \theta_2]}(S) := \min_{\theta \in [\theta_1, \theta_2]} \width_\theta(S).\]
Note that $\width(S) = \width_{[0,\pi]}(S)$
and $\width_{[\theta, \theta]}(S) = \width_\theta(S) = \width(\strip_\theta(S))$.
Also, define
 \[ \strip_{[\theta_1, \theta_2]}(S) := \bigcap_{\theta \in [\theta_1, \theta_2]} \strip_\theta(S).\]
Note that if $\theta_1\neq\theta_2$, $\strip_{[\theta_1, \theta_2]}(S)$ is
the convex hull of $S$ and two more points from the boundary of the intersection
of two strips~$\strip_{\theta_1}(S)$ and~$\strip_{\theta_2}(S)$.
See \figurename~\ref{fig:sigma_int} for an illustration.

\begin{figure}[tb]
  \centerline{\includegraphics[width=.45\textwidth]{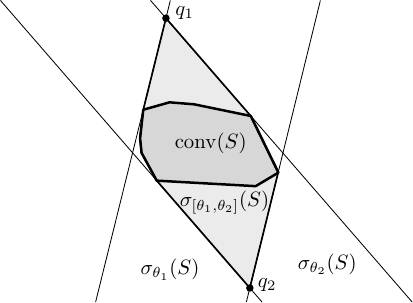}}
  \caption{Illustration of~$\strip_{[\theta_1,\theta_2]}(S) = \conv(S\cup\{q_1, q_2\})$ 
  (shaded in light gray) when $\theta_1 < \theta_2$.
  }
  \label{fig:sigma_int}
\end{figure}

\begin{lemma} \label{lem:constrained_width}
 For any finite set~$S$ of points, it holds that
  $ \width_{[\theta_1, \theta_2]}(S) = \width(\strip_{[\theta_1, \theta_2]}(S))$.
\end{lemma}
\begin{proof}
If $\theta_1 = \theta_2$, then the lemma is trivially true,
so we assume that $\theta_1 < \theta_2$.

Since $S \subseteq \strip_{[\theta_1, \theta_2]}(S)$,
we have
 \[\width_{[\theta_1, \theta_2]}(S)\leq\width_{[\theta_1, \theta_2]}(\strip_{[\theta_1, \theta_2]}(S)),\]
on one hand.
On the other hand, consider any strip~$\strip_\theta(S)$ that determines
$\width_{[\theta_1, \theta_2]}(S)$ for some $\theta \in [\theta_1, \theta_2]$.
Then, by definition, we have $\strip_{[\theta_1, \theta_2]}(S) \subseteq \strip_\theta(S)$.
This implies that
 \[\width_{[\theta_1, \theta_2]}(S)\geq\width_{[\theta_1, \theta_2]}(\strip_{[\theta_1, \theta_2]}(S)),\]
so we have the equality
 \[\width_{[\theta_1, \theta_2]}(S)=\width_{[\theta_1, \theta_2]}(\strip_{[\theta_1, \theta_2]}(S)).\]

Now, consider a strip~$\strip_\theta(\strip_{[\theta_1, \theta_2]}(S))$ that determines $\width(\strip_{[\theta_1, \theta_2]}(S))$.
We claim that $\theta \in [\theta_1, \theta_2]$.
If the claim is true, then we conclude
 \[ \width(\strip_{[\theta_1, \theta_2]}(S)) = \width_{[\theta_1,\theta_2]}(\strip_{[\theta_1, \theta_2]}(S)) = \width_{[\theta_1,\theta_2]}(S),\]
and the lemma is proved.

In order to prove the claim, suppose that $\theta \notin [\theta_1,\theta_2]$.
As discussed above and illustrated in \figurename~\ref{fig:sigma_int},
$\strip_{[\theta_1, \theta_2]}(S)$ is the convex hull of $S$ and two more points $q_1$ and $q_2$
that are two vertices of the boundary of 
the intersection $\strip_{\theta_1}(S) \cap \strip_{\theta_2}(S)$.
By the construction, 
the two bounding lines of $\strip_\theta(\strip_{[\theta_1, \theta_2]}(S))$ contains these two points $q_1$ and $q_2$
and nothing else from the original set~$S$ since $\theta \notin [\theta_1, \theta_2]$.
This leads to a contradiction to the fact that
at least three points of~$\strip_{[\theta_1, \theta_2]}(S)$ lie on the boundary of the minimum-width strip enclosing~$\strip_{[\theta_1, \theta_2]}(S)$.
The claim is thus proved.
\end{proof}

As will be seen later,
we are also interested in \emph{orientation-constrained width decision queries}.
More precisely,
we are given a query interval $[\theta_1,\theta_2] \subseteq [0,\pi)$ of orientations
and want to decide whether there exists $\theta \in [\theta_1,\theta_2]$
such that the width of $\strip_\theta(S)$ is at most~$\wparam$
or, equivalently, whether $\width_{[\theta_1, \theta_2]}(S) \leq \wparam$.
It turns out that the structure~$\ASDS$ by Agarwal and Sharir is helpful 
for this type of queries as well,
with the aid of Lemma~\ref{lem:constrained_width}.
\begin{lemma} \label{lem:constrained_width-query}
 Provided the data structure~$\ASDS$ on a point set~$S$ of $n$~points with parameter~$\wparam$ 
 is available,
 an orientation-constrained width decision query on~$S$ for width~$\wparam$ can be answered 
 in $O(\log^2 n)$ worst-case time using $O(\log n)$ additional space.
\end{lemma}
\begin{proof}
Let $[\theta_1, \theta_2]$ be a query interval 
for an orientation-constrained width decision query.
By Lemma~\ref{lem:constrained_width},
we know that $\width_{[\theta_1, \theta_2]}(S) = \width(\strip_{[\theta_1,\theta_2]}(S))$.

If $\theta_1=\theta_2$, then this can be done by
explicitly computing the strip~$\strip_{\theta_1}(S)$.
This can be done in $O(\log n)$ time by searching the internal binary search trees in~$\ASDS$.
So, we assume that $\theta_1 < \theta_2$.
As described above,
in this case, 
$\strip_{[\theta_1,\theta_2]}(S) = \conv(S \cup \{q_1, q_2\})$, where $q_1$ and $q_2$
are two vertices of the boundary of 
the intersection $\strip_{\theta_1}(S) \cap \strip_{\theta_2}(S)$.
(See \figurename~\ref{fig:sigma_int}.)
Thus, we have
 \[ \width_{[\theta_1, \theta_2]}(S) = 
  \width(\strip_{[\theta_1,\theta_2]}(S)) = \width(S \cup \{q_1, q_2\}).\]
In order to decide whether $\width_{[\theta_1, \theta_2]}(S) \leq \wparam$,
we insert two points $q_1$ and $q_2$ into~$S$ and update~$\ASDS$ accordingly.
This takes $O(\log^2 n)$ time by Lemma~\ref{lem:ASDS1}.
Afterwards, we delete $q_2$ and $q_1$ from~$S$ in this order to recover~$\ASDS$.
This can also be done by undoing the latest two insertions,
so we spend $O(\log^2 n)$ time again by Lemma~\ref{lem:ASDS1}.
In this process, we spend $O(\log n)$ additional space
to record at most $O(\log n)$ subtrees in~$\ASDS$ that represents the set of removed edges 
from~$\conv(S)$ by inserting $q_1$ and $q_2$,
which is freed right after undoing the insertions~\cite{as-odmwpps-91}.
\end{proof}

\begin{figure}[bt]
  \centerline{\includegraphics[width=.90\textwidth]{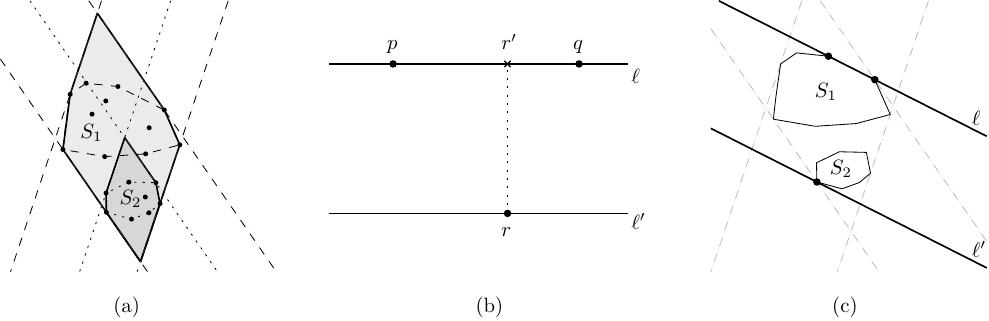}}
  \caption{Illustrations to the proof of Lemma~\ref{lem:constrained_width2}.
  }
  \label{fig:constrained_width2}
\end{figure}

Now, consider two sets $S_1$ and $S_2$ of points in the plane
that can be separated by a line, that is, $\conv(S_1) \cap \conv(S_2) = \emptyset$.
Then, there are exactly two outer common tangent lines~$\ell_1$ and~$\ell_2$.
Let $\theta_1 \leq \theta_2$ be the orientations of $\ell_1$ and $\ell_2$.
We say that $S_1$ \emph{dominates}~$S_2$
if $\strip_{[\theta_1,\theta_2]}(S_2) \subseteq \strip_{[\theta_1,\theta_2]}(S_1)$.
By construction, note that either $S_1$ or $S_2$ dominates the other.
We also mean by the distance between two convex, compact sets~$A$ and~$B$, denoted by~$d(A,B)$,
the minimum length of translation vectors~$\tau$ such that $A$ and $B+\tau$ have a common point.
\begin{lemma} \label{lem:constrained_width2}
 With the above notations, suppose that $S_1$ dominates~$S_2$.
 Then, it holds that:
 \begin{enumerate}[(i)] \denseitems
  \item $\width_{[\theta_1,\theta_2]}(S_1\cup S_2) = \width_{[\theta_1,\theta_2]}(S_1)$.
  \item If $\width(S_1 \cup S_2) < d(\conv(S_1), \conv(S_2))$, then
   $\width(S_1\cup S_2) = \width_{[\theta_1,\theta_2]}(S_1)$.
 \end{enumerate}
\end{lemma}
\begin{proof}
Since $S_1$ dominates~$S_2$,
we have
 \[S_2 \subseteq \strip_{[\theta_1,\theta_2]}(S_2) \subseteq \strip_{[\theta_1,\theta_2]}(S_1).\]
Lemma~\ref{lem:constrained_width} implies that
\[ \strip_{[\theta_1,\theta_2]}(S_1 \cup S_2) = \strip_{[\theta_1,\theta_2]}(S_1),\]
so the first statement~(i) follows.
See \figurename~\ref{fig:constrained_width2}(a).

Suppose $\width(S_1 \cup S_2) < d(\conv(S_1), \conv(S_2))$.
Let $\strip^* = \strip_{\theta^*}(S_1 \cup S_2)$ be 
a minimum-width strip enclosing $S_1 \cup S_2$ whose orientation is~$\theta^*$.
We claim that $\theta^* \in [\theta_1, \theta_2]$.
If this claim is true, then, by definition, we have 
 \[ \strip_{[\theta_1,\theta_2]}(S_1) = \strip_{[\theta_1,\theta_2]}(S_1 \cup S_2) \subseteq 
   \strip_{\theta^*}(S_1 \cup S_2) = \strip^*,\]
so
 \[ \width(\strip_{[\theta_1,\theta_2]}(S_1)) \leq \width(S_1 \cup S_2),\]
on one hand.
On the other hand, since $S_1\cup S_2$ is a subset of $\strip_{[\theta_1, \theta_2]}(S_1 \cup S_2) = \strip_{[\theta_1,\theta_2]}(S_1)$,
we also have 
 \[ \width(\strip_{[\theta_1,\theta_2]}(S_1)) \geq \width(S_1 \cup S_2),\]
and the second statement~(ii) is thus proved.

Hence, we are done by proving the claim that $\theta^* \in [\theta_1, \theta_2]$.
As $\strip^*$ is minimal among those enclosing $S_1\cup S_2$, 
its boundary contains three points $p, q, r$ from $S_1 \cup S_2$ such that
$p$ and $q$ lie on a common bounding line~$\ell$ of~$\strip^*$,
$r$ lies on the other bounding line~$\ell'$,
and the perpendicular foot~$r'$ of~$r$ to~$\ell$ lies in between~$p$ and~$q$.
See \figurename~\ref{fig:constrained_width2}(b) for an illustration.

We first exclude the possibility that both $p$ and $q$ belong to a common set, $S_1$ or $S_2$,
and $r$ to the other.
Suppose for a contradiction that, say, $p, q\in S_1$ and $r\in S_2$.
By our assumption that $d(\conv(S_1),\conv(S_2))>\width(S_1 \cup S_2)$,
the distance~$d$ from~$r$ to its perpendicular foot~$r'$ is strictly larger than~$\width(S_1 \cup S_2)$,
while, however, the distance~$d$ is also the width of~$\strip^*$,
a contradiction to the assumption that the width of~$\strip^*$ determines $\width(S_1 \cup S_2)$.

Now, suppose that $\theta^* \notin [\theta_1,\theta_2]$.
Then, observe that both bounding lines~$\ell$ and $\ell'$ of~$\strip^*$
cannot intersect a common set, $S_1$ or $S_2$;
that is, $\ell\cap S_i \neq \emptyset$ if and only if $\ell' \cap S_i = \emptyset$, for $i=1,2$.
(See \figurename~\ref{fig:constrained_width2}(c).)
This implies that $p, q\in \ell$ belong to one common set, $S_1$ or $S_2$,
and $r\in \ell'$ belongs to the other set,
which is forbidden by the above argument.
Thus, we have $\theta \in [\theta_1, \theta_2]$,
and the claim is true.
\end{proof}

\subsection{Decision algorithm}
We describe our decision algorithm
for a given parameter~$\wparam>0$.
The points $p_1, \ldots, p_n \in P$ are assumed to be sorted in the $y$-coordinates,
and we let $P_i = \{p_1, \ldots, p_i\}$ and $\overline{P}_i = P\setminus P_i$.

\begin{figure}[t]
  \centerline{\includegraphics[width=.65\textwidth]{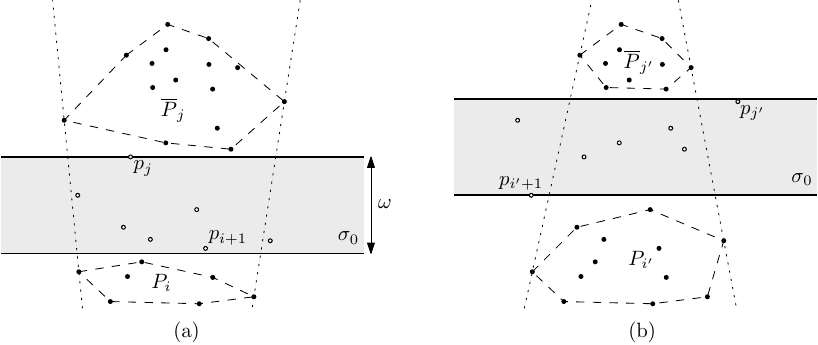}}
  \caption{Snapshots of the sweeping process: 
  (a) $\overline{P}_j$ dominates~$P_i$ and (b) $P_{i'}$ dominates~$\overline{P}_{j'}$.}
  \label{fig:1fixed_decision}
\end{figure}

Consider a horizontal strip~$\strip_0$ of width~$\wparam$
and sweep the plane by translating~$\strip_0$ upwards from below.
At any moment of this sweeping process, the points $P\setminus \strip_0$ outside of~$\strip_0$
is partitioned into $P_i$ and $\overline{P}_j$ for some $0\leq i\leq j\leq n$
such that all points of~$P_i$ lies below~$\strip_0$ and 
all points of~$\overline{P}_j$ lies above~$\strip_0$.
Note that the indices~$i$ and $j$ representing such a separation by~$\strip_0$
do not decrease during the process,
and it always holds that $d(\conv(P_i), \conv(\overline{P}_j)) > \wparam$.
Also, by this monotonicity of $i$ and $j$,
observe that $\overline{P}_j$ dominates $P_i$ from the beginning until some moment
and $P_i$ dominates $\overline{P}_j$ from that moment to the end.
See \figurename~\ref{fig:1fixed_decision}.

We maintain convex hulls, $\conv(P_i)$ and $\conv(\overline{P}_j)$,
and the data structure~$\ASDS$ with parameter~$\wparam$ on set~$P_i$.
Maintaining the convex hulls $\conv(P_i)$ and $\conv(\overline{P}_j)$
can be done in $O(n\log n)$ total time~\cite{bj-dpch-02, hs-ompc-91};
in our case, updates on $P_i$ and $\overline{P}_j$ are offline.
Initially, we have $i=j=0$,
so $\conv(P_i) = \emptyset$, $\conv(\overline{P}_j) = \conv(P)$, and $\ASDS$ is initialized for an empty set~$P_0$.

In the main loop of our decision algorithm,
as $\strip_0$ moves upwards,
we do nothing until $P_i$ becomes dominating~$\overline{P}_j$
while we maintain the data structure~$\ASDS$ by inserting relevant points.
For each pair $(i, j)$ such that $P_i$ dominates~$\overline{P}_j$, 
we decide whether $\width(P_i \cup \overline{P}_j) \leq \wparam$ or not.
If it is the case, we stop the algorithm and report \texttt{YES}; otherwise, we proceed the algorithm.
The decision is made as follows:
We first compute the outer common tangents of $\conv(P_i)$ and $\conv(\overline{P}_j)$ 
in $O(\log^2 n)$ time
using $\conv(P_i)$ and $\conv(\overline{P}_j)$. (See \figurename~\ref{fig:1fixed_decision}.)
Let $\theta_1 \leq \theta_2$ be the orientations of the two common tangents.
We then perform an orientation-constrained width decision query on~$P_i$
with query interval~$[\theta_1,\theta_2]$.
This can be done in $O(\log^2 n)$ time by Lemma~\ref{lem:constrained_width-query} using~$\ASDS$.
If the answer to this query is positive, then we conclude that $\width(P_i \cup \overline{P}_j) \leq \wparam$;
otherwise, we conclude that $\width(P_i \cup \overline{P}_j) > \wparam$.
The correctness of this decision is guaranteed 
by the following lemma, which is a direct application of Lemma~\ref{lem:constrained_width2}.
\begin{lemma} \label{lem:1fixed_decision}
 For $(i,j)$ such that $P_i$ dominates $\overline{P}_j$ and $d(\conv(P_i), \conv(\overline{P}_j)) > \wparam$,
 we have
 $\width(P_i \cup \overline{P}_j) \leq \wparam$ if and only if 
 $\width_{[\theta_1, \theta_2]}(P_i) \leq \wparam$.
\end{lemma}
\begin{proof}
Recall that $P_i$ and $\overline{P}_j$ are separated by the horizontal strip~$\strip_0$ of width~$\wparam$,
so we have $d(\conv(P_i), \conv(\overline{P}_j)) > \wparam$.
Also, note that 
 \[ \width(P_i \cup \overline{P}_j) \leq \width_{[\theta_1, \theta_2]}(P_i \cup \overline{P}_j) 
     = \width_{[\theta_1, \theta_2]}(P_i),\]
in general by Lemma~\ref{lem:constrained_width2}(i).
Thus,
if $\width_{[\theta_1, \theta_2]}(P_i) \leq \wparam$,
then we have
 \[ \width(P_i \cup \overline{P}_j) \leq \width_{[\theta_1, \theta_2]}(P_i) \leq \wparam.\]

If $\width(P_i\cup \overline{P}_j) \leq \wparam$, 
then we have $\width(P_i \cup \overline{P}_j) \leq \wparam < d(\conv(P_i), \conv(\overline{P}_j))$.
Hence, Lemma~\ref{lem:constrained_width2}(ii) implies that
 \[\width_{[\theta_1, \theta_2]}(P_i) = \width(P_i\cup \overline{P}_j)  \leq \wparam,\]
since $P_i$ dominates~$\overline{P}_j$.
\end{proof}

This way, we check all the pairs $(i, j)$ such that $P_i$ dominates~$\overline{P}_j$ 
during the sweeping process.
The other pairs~$(i, j)$ such that $\overline{P}_j$ dominates~$P_i$ can be handled
in a symmetric way by moving the horizontal strip~$\strip_0$ downwards.
Therefore, we conclude the following.
\begin{theorem} \label{thm:1fixed_decision}
 Given a set~$P$ of $n$~points, $\phi \in [0,\pi)$, and $\wparam>0$,
 we can decide in $O(n\log^2 n)$ time and $O(n)$ space
 whether there is a constrained two-strip of width~$\wparam$.
\end{theorem}
\begin{proof}
The correctness of our algorithm is shown by Lemma~\ref{lem:1fixed_decision}
as discussed above.

The complexity of our algorithm is analyzed as follows.
We spend $O(n\log n)$ time to sort~$P$ in the beginning.
Maintaining $\conv(P_i)$ and $\conv(\overline{P}_j)$ costs $O(n \log n)$ in total~\cite{hs-ompc-91}.
The total cost for~$\ASDS$ is bounded by $O(n \log^2 n)$ by Lemma~\ref{lem:ASDS}.
For each pair $(i,j)$ appearing in the sweeping process,
we spend $O(\log^2 n)$ time to compute the common tangents of $\conv(P_i)$ and $\conv(\overline{P}_j)$
and $O(\log^2 n)$ time for an orientation-constrained width decision query 
by Lemma~\ref{lem:constrained_width-query}.
Since there are at most $2n$ such pairs $(i,j)$,
the total time is bounded by $O(n \log^2 n)$ in the worst case.
The space used is $O(n)$.
\end{proof}

\subsection{Optimization algorithm}

Let $w^*$ be our target optimal width, that is,
the minimum width of a constrained two-strip enclosing~$P$.
As above, suppose that $p_1, \ldots, p_n\in P$ are sorted in their $y$-coordinates.
Let $W_1$ be the set of all differences of $y$-coordinates between two points in~$P$.
Since $W_1$ can be represented by a sorted matrix~\cite{fj-csrX+Ymsc-82},
we can find two consecutive values $w_0, w_1 \in W_1$ such that $w_0 < w^* \leq w_1$
in $O(n \log^3 n)$ time using our decision algorithm (Theorem~\ref{thm:1fixed_decision})
and an efficient selection algorithm for a sorted matrix~\cite{fj-gsrsm-84}.

At this stage, observe that, for any $w_0 < w < w_1$,
the sequence of changes on the sets $P_i$ and $\overline{P}_j$ of points below and above~$\strip_0$ is the same
as the horizontal strip~$\strip_0$ of width~$w$ moves upwards.
Let $X$ be the set of those pairs $(i, j)$ of indices
such that $P_i$ and $\overline{P}_j$ appear as the sets of points below and above~$\strip_0$,
respectively.
\begin{lemma} \label{lem:1fixed_opt_width}
 It holds that
  \[ w^* = \min\{w_1, \min_{(i,j)\in X} \{\width(P_i \cup \overline{P}_j)\}\}.\]
\end{lemma}
\begin{proof}
First, observe that 
 \[ w^* = \min\{w_1, \min_{(i,j)\in X}\{\max\{\width_0(P\setminus (P_i \cup \overline{P}_j)),
    \width(P_i \cup \overline{P}_j)\}\}\}.\]
Note that 
$\width_0(P\setminus (P_i \cup \overline{P}_j)) = \width(\strip_0(P\setminus (P_i \cup \overline{P}_j)))$
is the width of the smallest horizontal strip enclosing 
$P\setminus (P_i \cup \overline{P}_j) = \{p_{i+1},\ldots, p_{j}\}$,
so it is determined by the $y$-difference of~$p_j$ and~$p_{i+1}$.
Since $w_0$ and $w_1$ are two consecutive values in~$W_1$ and $w_0 < w < w_1$, 
this implies that $\width_0(P\setminus (P_i \cup \overline{P}_j)) \leq w_0$
for any~$(i,j)\in X$.

If $\width(P_i \cup \overline{P}_j) \leq \width_0(P\setminus (P_i \cup \overline{P}_j))$ 
for some $(i,j)\in X$,
then we would conclude that $w^* \leq \width_0(P\setminus (P_i \cup \overline{P}_j)) \leq w_0$,
a contradiction to the fact that $w_0 < w^*$.
Therefore, we have $\width(P_i \cup \overline{P}_j) > \width_0(P\setminus (P_i \cup \overline{P}_j))$ 
for all $(i,j)\in X$,
so the lemma follows.
\end{proof}

Thus, we are done by computing the width of $P_i \cup \overline{P}_j$ for $(i,j)\in X$.
Even better, to compute~$w^*$ and an optimal two-strip,
it suffices to evaluate the exact value of $\width(P_i\cup \overline{P}_j)$
only for those $(i,j)\in X$ such that $\width(P_i\cup \overline{P}_j) < w_1$.
Let $W$
be the set of values of~$\width(P_i \cup \overline{P}_j)$
that are less than~$w_1$.
Below, we show how to compute the set~$W$.

For the purpose,
we take any value $w$ with $w_0 < w <w_1$
and simulate the translational sweeping process with the horizontal strip~$\strip_0$ of width~$w$
in a similar way as done for the decision algorithm.
Here, we sweep the plane by moving~$\strip_0$ downwards from above.
We initialize the data structure~$\ASDS$ with parameter~$\wparam = w_1$ on point set~$P_n = P$
by inserting $n$~points $p_1, \ldots, p_n$ in this order,
and maintain~$\ASDS$ by deleting points in the reversed order to represent~$P_i$,
as~$\strip_0$ moves downwards.
In addition, we initialize the structure~$\CWDS$ of Lemma~\ref{lem:Chan_width}
for point set~$P_n = P$, and maintain it to store~$P_i$
by deleting points in the same order.
We also maintain the convex hulls $\conv(P_i)$ and $\conv(\overline{P}_j)$.

During the sweeping process,
we only handle those $(i,j)\in X$ such that $P_i$ dominates $\overline{P}_j$,
so we stop the process as soon as $\overline{P}_j$ dominates $P_i$.
(Those~$(i,j)\in X$ such that $\overline{P}_j$ dominates $P_i$
 can be handled in a symmetric way by moving $\strip_0$ upwards.)
Consider such a pair $(i, j)$.
Our goal is to compute $\width(P_i \cup \overline{P}_j)$ only when it is less than~$w_1$.
Note that $d(\conv(P_i), \conv(\overline{P}_j)) \geq w_1$.
We compute the outer common tangents of $\conv(P_i)$ and $\conv(\overline{P}_j)$
and let $\theta_1\leq \theta_2$ be their orientations.
As in the decision algorithm,
we test whether or not $\width_{[\theta_1,\theta_2]}(P_i) \leq w_1$
by Lemma~\ref{lem:constrained_width-query} using~$\ASDS$.
Lemma~\ref{lem:constrained_width2} implies the following.
\begin{lemma} \label{lem:1fixed_opt1}
 Provided $P_i$ dominates~$\overline{P}_j$,
 $\width(P_i \cup \overline{P}_j) \geq w_1$ if $\width_{[\theta_1,\theta_2]}(P_i) > w_1$.
\end{lemma}
\begin{proof}
Recall that $d(\conv(P_i), \conv(\overline{P}_j)) \geq w_1$.
Suppose that $\width(P_i \cup \overline{P}_j) < w_1$ while $\width_{[\theta_1,\theta_2]}(P_i) > w_1$.
Then, we have
 \[ \width(P_i \cup \overline{P}_j) < w_1 \leq d(\conv(P_i), \conv(\overline{P}_j)).\]
So, Lemma~\ref{lem:constrained_width2} implies
 \[ \width(P_i \cup \overline{P}_j) = \width_{[\theta_1,\theta_2]}(P_i),\]
since $P_i$ dominates~$\overline{P}_j$, a contradiction.
\end{proof}

If it turns out that $\width_{[\theta_1,\theta_2]}(P_i) > w_1$, then
we can discard the pair~$(i,j)$ and proceed the algorithm by Lemma~\ref{lem:1fixed_opt1}.
Otherwise, we compute the exact value of $\width_{[\theta_1,\theta_2]}(P_i)$.
If $\theta_1 = \theta_2$, this can be done in $O(\log n)$ time using convex hulls
$\conv(P_i)$ and $\conv(\overline{P}_j)$;
if $\theta_1 < \theta_2$, by Lemma~\ref{lem:constrained_width},
we have $\width_{[\theta_1,\theta_2]}(P_i) = \width(\strip_{[\theta_1,\theta_2]}(P_i))$
and $\strip_{[\theta_1,\theta_2]}(P_i) = \conv(P_i \cup Q)$
where $Q = \{q_1, q_2\}$ consists of two points as described above.
(See also \figurename~\ref{fig:sigma_int}.)
Hence, in this case, we can compute $\width_{[\theta_1,\theta_2]}(P_i)$
in $O(\log^3 n)$ time by Lemma~\ref{lem:Chan_width}
with a query set~$Q = \{q_1, q_2\}$ to~$\CWDS$.
Again, Lemma~\ref{lem:constrained_width2} implies the following.
\begin{lemma} \label{lem:1fixed_opt2}
 $\width_{[\theta_1,\theta_2]}(P_i) < w_1$ if and only if
 $\width(P_i\cup \overline{P}_j) < w_1$.
 Moreover, if $\width_{[\theta_1,\theta_2]}(P_i) < w_1$, 
 then $\width(P_i\cup \overline{P}_j) = \width_{[\theta_1,\theta_2]}(P_i)$.
\end{lemma}
\begin{proof}
Note that $d(\conv(P_i), \conv(\overline{P}_j)) \geq w_1$,
and also that $\overline{P}_j \subset \strip_{[\theta_1,\theta_2]}(P_i)$,
so we have 
 \[ \width(P_i \cup \overline{P}_j) \leq \width_{[\theta_1,\theta_2]}(P_i \cup \overline{P}_j)
   = \width_{[\theta_1,\theta_2]}(P_i) \leq w_1\]
by the assumption and Lemma~\ref{lem:constrained_width2}(i).

If $\width_{[\theta_1,\theta_2]}(P_i) < w_1$,
we have
 \[ \width(P_i \cup \overline{P}_j) \leq \width_{[\theta_1,\theta_2]}(P_i) < w_1 \leq d(\conv(P_i), \conv(\overline{P}_j)).\]
So, Lemma~\ref{lem:constrained_width2}(ii) implies that
\[ \width(P_i \cup \overline{P}_j) = \width_{[\theta_1,\theta_2]}(P_i).\]

Next, suppose that $\width(P_i \cup \overline{P}_j) < w_1$.
Then, Lemma~\ref{lem:constrained_width2}(ii) again implies that
 \[ \width(P_i \cup \overline{P}_j) = \width(\strip_{[\theta_1, \theta_2]}(P_i))  < w_1.\]
Thus, the lemma follows.
\end{proof}
Hence, we have $\width(P_i \cup \overline{P}_j) = \width_{[\theta_1,\theta_2]}(P_i)$
and it is a member of~$W$
if and only if 
the computed value of~$\width_{[\theta_1,\theta_2]}(P_i)$ is strictly smaller than~$w_1$.

This way, we can collect the values in~$W$ in $O(n \log^3 n)$ time.
By Lemma~\ref{lem:1fixed_opt_width},
the minimum value in~$W$ is~$w^*$, if $W$ is nonempty;
or $w^* = w_1$, if $W = \emptyset$.
The corresponding two-strip of width~$w^*$ can be computed and stored 
during the execution of the algorithm.
Hence, we finally conclude the following.
\begin{theorem}\label{thm:1fixed}
 Given a set~$P$ of $n$~points in the plane and an orientation~$\phi$,
 a two-line center for~$P$ in which one of the two lines is constrained to be in orientation~$\phi$
 can be computed in $O(n\log^3 n)$ time and $O(n\log n)$ space.
\end{theorem}
\begin{proof}
The correctness of our algorithm is established by
the above discussions together with 
Lemmas~\ref{lem:1fixed_opt_width}--\ref{lem:1fixed_opt2}.

To see the time complexity,
first we spend $O(n\log^3 n)$ time to compute the two values $w_0<w_1 \in W_1$,
using our $O(n\log^2 n)$-time decision algorithm, described in Theorem~\ref{thm:1fixed_decision}.
After specifying the set~$X$ of pairs of indices,
we perform the sweeping process,
in which we spend $O(n\log^3 n)$ total time for maintaining the data structures
$\conv(P_i)$, $\conv(\overline{P}_j)$, $\ASDS$, and $\CWDS$
by Lemmas~\ref{lem:ASDS} and~\ref{lem:Chan_width}.
For each $(i,j)\in X$,
we compute the outer common tangents to $\conv(P_i)$ and $\conv(\overline{P}_j)$ in $O(\log^2 n)$ time
and perform two queries:
an orientation-constrained width decision query in $O(\log^2 n)$ time 
by Lemma~\ref{lem:constrained_width-query}
and a width query to compute $\width(P_i \cup Q)$ by a query set~$Q$ of two points
in $O(\log^3 n)$ time by Lemma~\ref{lem:Chan_width}.
In total, the running time of our algorithm is bounded by $O(n\log^3 n)$.

The space used is bounded by $O(n \log n)$.
Remark that only the structure~$\CWDS$ requires $O(n \log n)$ space,
while $O(n)$ space is sufficient for the other structures.
\end{proof}

\section{Fixed angle of intersection} \label{sec:afixed}

In this section, we solve the third constrained two-line center problem
in which, given a real value~$0 \leq \beta \leq \pi/2$,
the difference of the orientations of the two resulting lines is exactly~$\beta$.
Throughout this section, for convenience,
a \emph{constrained two-strip} denotes
a pair of strips whose orientations differ by~$\beta$.
Let $w^*$ be the minimum width of a constrained two-strip enclosing~$P$.
We start by describing optimal configurations.
\begin{lemma} \label{lem:afixed_conf}
 There exists a minimum-width constrained two-strip $(\strip_1, \strip_2)$ enclosing~$P$
 that falls into one of the following cases:
 \begin{enumerate}[(i)] \denseitems
  \item Either $w^* = \width(\strip_1) > \width(\strip_2)$ 
  or $w^* = \width(\strip_2) > \width(\strip_1)$, and
  one of the four bounding lines of~$\strip_1$ and~$\strip_2$ contains two points in~$P$.
  \item It holds that $w^* = \width(\strip_1) = \width(\strip_2)$, and
  each of the four bounding lines of~$\strip_1$ and~$\strip_2$ contains
  a point in~$P$.
 \end{enumerate}
\end{lemma}
\begin{proof}
It is obvious that there exists a minimum-width constrained two-strip enclosing~$P$ such that
each of the four bounding lines contains at least one point of~$P$.
Let $(\strip_1,\strip_2)$ be such a two-strip.
If $\width(\strip_1)=\width(\strip_2)$, then it is case~(ii).
Thus, in the following, we assume that $w^* = \width(\strip_1)>\width(\strip_2)$.
The other case, $w^* = \width(\strip_2)>\width(\strip_1)$, is symmetric.

In order to show that this is case~(i),
suppose each of the bounding line of~$\strip_1$ contains exactly one point of~$P$.
In other words, $\strip_1$ is determined by two points~$q_1, q_2 \in P$
that lie on its boundary, that is, $\strip_1 = \strip_\theta(\{q_1, q_2\})$
for some~$\theta \in [0,\pi)$.
Note that its width~$\width_\theta(\{q_1, q_2\})$ is a sinusoidal function of~$\theta$,
which is of the form~$A\sin(\theta + B)$~\cite{Bae2020}.
Since $\width(\strip_1) > \width(\strip_2) \geq 0$,
there always exist a real~$\epsilon$, positive or negative, whose absolute value is sufficiently small
such that
  $\strip_{\theta+\epsilon}(P\cap \strip_1) = \strip_{\theta+\epsilon}(\{q_1, q_2\})$
and
 \[ \width_{\theta'+\epsilon}(P\cap \strip_2) < \width_{\theta+\epsilon}(\{q_1, q_2\}) < \width(\strip_1) = w^*,\]
a contradiction to the optimality of~$(\strip_1,\strip_2)$.
\end{proof}

In the following, we present an algorithm that runs in $O(n^2 \alpha(n) \log n)$ time 
using $O(n^2)$ space.
Our algorithm follows a similar flow as for the second problem,
consisting of two phases:
(1) find a favorably narrow interval~$(w_0, w_1]$ that includes our target width~$w^*$,
and (2) proceed the search for~$w^*$ with the aid of~$(w_0, w_1]$.
We first present an efficient decision algorithm,
and then describe each of the two phases.

\subsection{Decision algorithm} \label{sec:afixed_decision}

Let $\wparam>0$ be a given parameter, and our goal is to decide whether or not $\wparam\geq w^*$.
Throughout this section, we regard $\theta$ as a \emph{direction} from the range~$[0, 2\pi)$,
taken by modulo~$2\pi$, and assume that no three points in~$P$ are collinear.

\begin{figure}[b]
  \centerline{\includegraphics[width=.45\textwidth]{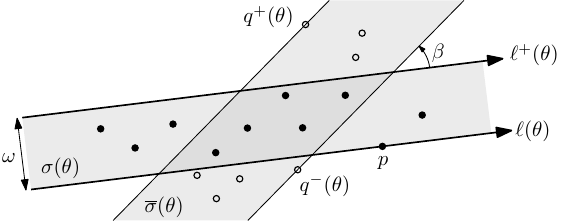}}
  \caption{A snapshot at~$\theta$ of the rotational sweeping process with fixed width~$\wparam$
  and pivot~$p$.}
  \label{fig:sweep}
\end{figure}

We consider a rotational sweeping process with \emph{fixed width~$\wparam$} described as follows:
Take any point~$p\in P$ as a pivot.
For direction~$\theta \in [0,2\pi)$, 
let $\ell(\theta)$ and $\ell^+(\theta)$ be two directed lines in~$\theta$
such that $\ell(\theta)$ is through~$p$
and $\ell^+(\theta)$ is at distance~$\wparam$ to the left of~$\ell(\theta)$.
We simultaneously rotate both lines~$\ell(\theta)$ and~$\ell^+(\theta)$ counterclockwise 
by increasing~$\theta$,
and consider the strip~$\strip(\theta)$ bounded by the two lines.
Taking the second strip~$\stripbar(\theta):=\strip_{\theta+\beta}(P\setminus \strip(\theta))$
as the minimum-width strip in orientation~$\theta+\beta$ enclosing the rest of points in~$P$,
our goal is to decide if there exists $\theta \in [0,2\pi)$ such that 
$\width(\stripbar(\theta)) \leq \wparam$ for some $p\in P$.
See \figurename~\ref{fig:sweep},
in which points in~$P\cap \strip(\theta)$ are depicted by dots
and those in~$P\setminus \strip(\theta)$ by small circles.

This sweeping process can be simulated
by maintaining the dynamic convex hull~$\conv(P\setminus \strip(\theta))$ and
its two extreme points~$q^-(\theta)$ and~$q^+(\theta)$ that define~$\stripbar(\theta)$.
Since the number of updates on~$P\setminus \strip(\theta)$ is~$O(n)$,
it can be done in~$O((n+E)\log n)$ time~\cite{bj-dpch-02},
where $E$ denotes the number of changes of the two extreme points~$q^-(\theta)$ and~$q^+(\theta)$.
As will be seen later, $E = O(n\alpha(n))$ and hence
the decision can be made in $O(n^2 \alpha(n) \log n)$ time.

In order to see why $E=O(n\alpha(n))$ and even to improve the running time,
we find it more useful and convenient to discuss the problem in the dual setting.
Consider the standard dual transformation that maps each point~$r=(a, b)\in \Plane$ 
into a non-vertical line $\dual{r} \colon \{y=ax-b\}$, and vice versa.
Let $L:=\{\dual{p} \mid p\in P\}$ be the set of $n$~lines dual to each point in~$P$.
For a fixed pivot~$p=(a,b)\in P$,
the trace of~$\ell^+(\theta)$ in the dual environment draws 
a hyperbola~$\{y= ax-b \pm \wparam\sqrt{1+x^2}\}$~\cite{as-pglp-94}.
We take the upper branch of the hyperbola, denoted by
 \[ h_p \colon \{y= ax-b + \wparam\sqrt{1+x^2}\},\]
and let $H:= \{h_p \mid p\in P\}$.
By this choice, we restrict ourselves to considering 
half the domain of directions, namely~$[\pi/2, 3\pi/2)$;
the other case can be handled symmetrically by considering the lower branches of the hyperbolas.
Note that the dual of~$\ell^+(\theta)$ for $\theta\in(\pi/2, 3\pi/2)$
is the intersection point between~$h_p$ and vertical line~$\{x=\tan(\theta)\}$.
Hence, the first strip~$\strip(\theta)$ appears as
a vertical segment between $p^*$ and $h_p$ at $x = \tan(\theta)$.
Similarly, the dual of the second strip~$\stripbar(\theta)$
is a vertical segment at $x = \tan(\theta + \beta)$
that crosses all but those lines in~$L$ intersected by~$\dual{(\strip(\theta))}$.

\begin{figure}[tb]
  \centerline{\includegraphics[width=.80\textwidth]{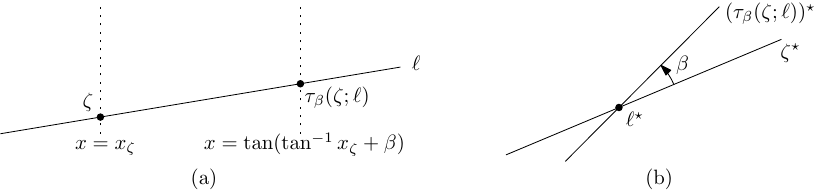}}
  \caption{(a) Illustration for the mapping~$\tau_\beta(\cdot; \ell)$ on line~$\ell$
  and (b) its dual representation.}
  \label{fig:tau}
\end{figure}

For better exposition,
we consider an operation that rotates a given line around a given point by angle~$\beta$.
In the dual setting, such a rotation can be performed by the following mapping:
for any non-vertical line~$\ell$ and a point~$\zeta \in \ell$, 
we define $\tau_\beta(\zeta; \ell)$ to be the intersection point between $\ell$ and
the vertical line $\{x = \tan(\tan^{-1}(x_\zeta) + \beta)\}$,
where $x_\zeta$ is the $x$-coordinate of~$\zeta$.
See \figurename~\ref{fig:tau}.
For a line segment~$s$ on~$\ell$, 
let $\tau_\beta(s)$ be the segment of~$\ell$ obtained by applying the $\beta$-shifting map~$\tau_\beta$
to all points on~$s$ along~$\ell$, 
that is, $\tau_\beta(s) = \bigcup_{\zeta\in s} \tau_\beta(\zeta;\ell)$.
Note that $\tau_\beta(s)$ may consist of two half-lines, if the $x$-coordinates of $s$ are large enough.

Now, consider the pieces of lines in~$L$ below $\dual{p}$ or above~$h_p$.
Let $S$ be the set of these segments and half-lines,
and $T := \{\tau_\beta(s) \mid s\in S\}$ be the set of shifted segments.
We then observe that the two lines bounding~$\stripbar(\theta)$
correspond to the highest and lowest points of the intersection~$T \cap \{x = \tan(\theta+\beta)\}$.
Thus, $\stripbar(\theta)$ and its width over $\theta \in (\pi/2, 3\pi/2)$ are determined by
the lower and upper envelopes, $\LE(T)$ and $\UE(T)$, of~$T$.
This implies that the number~$E$ of changes of the two extreme points~$q^-(\theta)$ and~$q^+(\theta)$
indeed counts the number of vertices in~$\LE(T)$ and $\UE(T)$,
so $E = O(n\alpha(n))$~\cite{sa-dsstga-95} 
and the decision problem can be solved in $O(n^2\alpha(n)\log n)$ time.
In the following, we improve it to $O(n^2 \alpha(n))$ time.

We start with the directed line~$\ell(\pi/2)$ through any~$p\in P$,
and rotate it around the pivot~$p$ by increasing~$\theta$
until it hits another point~$p'\in P$.
Whenever $\ell(\theta)$ hits another point~$p'$,
we switch the pivot to~$p'$ and continue the rotation around the new pivot~$p'$.
Observe that this motion of~$\ell(\theta)$ preserves the number~$k$
of points in~$P$ that lie on~$\ell(\theta)$ or on its right side,
except some moments when $\ell(\theta)$ contains two points.
In the dual setting, the trace of~$\ell(\theta)$ is well known as
the $k$-level of the arrangement~$\arr(L)$ of lines in~$L$.
More precisely, for $k=1,\ldots, n$,
the \emph{$k$-level} of~$\arr(L)$, denoted by~$L_k$,
is the monotone chain consisting of all edges~$e$ of~$\arr(L)$ such that
there are exactly $k-1$~lines strictly below any point in the relative interior of~$e$.
Similarly,
the trace of line~$\ell^+(\theta)$ at distance~$\wparam$ from~$\ell(\theta)$
is the $k$-level of the arrangement~$\arr(H)$ of $n$~hyperbolas in~$H$,
denoted by~$H_k$.

Let $L_k^-$ be the region strictly below~$L_k$
and $H_k^+$ be that strictly above~$H_k$.
For each~$r\in P$ and $1\leq k\leq n$, define
 \[ S^+_{k,r} := \dual{r}\cap H_k^+, \quad 
    S^-_{k,r} := \dual{r}\cap L_k^-, \quad 
    T^+_{k,r} := \tau_\beta(S^+_{k,r}), \quad \text{and} \quad
    T^-_{k,r} := \tau_\beta(S^-_{k,r}),\]
and let $T^+_k$ and $T^-_k$ be the collections of segments and half-lines
in $T^+_{k,r}$ and $T^-_{k,r}$, respectively, over all $r\in P$.
Our goal is then to compute the envelopes
$\LE(T^+_k \cup T^-_k)$ and $\UE(T^+_k \cup T^-_k)$ for all~$k$.
As discussed above,
these envelopes explicitly describe the changes of the extreme points defining
the second strip~$\stripbar(\theta)$ whose orientation is~$\theta + \beta$,
so we can decide whether $\width(\stripbar(\theta)) \leq \wparam$ for some $\theta$
in time linear to the total complexity of the envelopes.

Let $\LE^+_k := \LE(T^+_k)$, $\UE^+_k := \UE(T^+_k)$, $\LE^-_k := \LE(T^-_k)$, 
and $\UE^-_k := \UE(T^-_k)$.
We compute these four families of envelopes separately for all~$k$.
Then, $\LE(T^+_k \cup T^-_k)$ and $\UE(T^+_k \cup T^-_k)$ can be obtained
by merging two envelopes.
For our purpose, the following observation is essential.
\begin{lemma}\label{lem:afixed_decision_incremental}
 For any $1\leq k \leq n-1$ and $r\in P$,
  \[ T^+_{k+1,r} \subset T^+_{k,r} \quad \text{and} \quad T^-_{k,r} \subset T^-_{k+1,r}.\]
 Therefore, every point on~$\LE^+_k$ is below or on~$\LE^+_{k+1}$;
 every point on~$\UE^+_k$ is above or on~$\UE^+_{k+1}$;
 every point on~$\LE^-_k$ is above or on~$\LE^-_{k+1}$;
 every point on~$\UE^-_k$ is below or on~$\UE^-_{k+1}$.
\end{lemma}
\begin{proof}
Observe that
 \[ H_{k+1}^+ \subset H_k^+ \quad \text{and} \quad L_k^- \subset L_{k+1}^-\]
by definition, and thus that
 \[ S^+_{k+1,r} \subset S^+_{k,r} \quad \text{and} \quad S^-_{k,r} \subset S^-_{k+1,r}\]
for each $r\in P$.
Since the map~$\tau_\beta$ is continuous and bijective on any line~$\ell$,
when it is applied to segments on~$\ell$,
the lemma is implied.
\end{proof}
This naturally suggests us computing each family of envelopes in an incremental way.
In the following, we describe how to compute $\cL^+_n, \cL^+_{n-1}, \ldots, \cL^+_1$ in this order.
The other three families can also be handled symmetrically and analogously.

Initially,
we compute~$\arr(L)$ and $\arr(L\cup H)$,
and for each vertex~$v$ of~$\arr(L\cup H)$, we collect its images under~$\tau_\beta$
into a set~$\Xi$.
More precisely,
we use a dictionary structure for~$\Xi$ indexed by pairs in~$L\times(L\cup H)$,
such as an array-based $n \times 2n$ matrix.
For each pair $(\ell, \ell')$ with $\ell, \ell' \in L$ such that 
$v = \ell \cap \ell'$ is a vertex of~$\arr(L\cup H)$,
we store $\tau_\beta(v; \ell)$;
for $(\ell, h)$ with $\ell \in L$ and $h\in H$,
we store $\tau_\beta(v; \ell)$ for each $v\in \ell \cap h$.
We also associate each point $\xi\in \Xi$ with the edge $e$ of $\arr(L)$
such that $\xi \in e$,
so that given a pair in $L\times(L\cup H)$ we can locate in $O(1)$ time on which edges of~$\arr(L)$ 
its relevant points~$\xi \in \Xi$ lie.
Note that at most two points are stored at each entry of~$\Xi$.
The initialization can be done in $O(n^2)$ time using $O(n^2)$ space.

Let $T^+_{n+1} = \LE^+_{n+1} = \emptyset$, and
suppose $\LE^+_{k+1}$ has been correctly computed. 
Let $T^*$ be the set of segments obtained from $T^+_{k,r} \setminus T^+_{k+1,r}$
for all $r\in P$.
We then have $\LE^+_k = \LE(\LE^+_{k+1} \cup T^*)$ by
Lemma~\ref{lem:afixed_decision_incremental},
so can be computed by merging $\LE^+_{k+1}$ and $\LE(T^*)$.
Computing~$\LE(T^*)$ is done in three steps:
specify~$T^*$,
compute the arrangement~$\arr(T^*)$,
and extract~$\LE(T^*)$ from~$\arr(T^*)$.

To specift~$T^*$,
we walk along~$H_k$ and~$H_{k+1}$ in~$\arr(L\cup H)$
and find out all intersections~$H_k \cap \dual{r}$ and $H_{k+1} \cap \dual{r}$ for each $r\in P$.
We are then able to extract all segments of~$\dual{r}$ that lie in between $H_k$ and $H_{k+1}$.
For each such segment~$s$, $\tau_\beta(s)$ is a member of~$T^*$.
This can be done in $O(m_k + m_{k+1} + n)$ time,
where $m_i$ denotes the number of vertices of~$\arr(L\cup H)$ along~$H_i$.
Note that the number of segments in~$T^*$ is at most $m_k + m_{k+1}$,
since the endpoints of their preimages under the $\beta$-shifting map~$\tau_\beta$
are all from the vertices along~$H_k$ and $H_{k+1}$.

Note that every $t\in T^*$ is a segment of a line in~$L$,
so $\arr(T^*)$ is a clipped portion of the entire arrangement~$\arr(L)$.
In addition, the endpoints of~$t$ are members of~$\Xi$,
so we can find out their exact locations in~$\arr(L)$ in $O(1)$ time per each.
Thus, we can construct $\arr(T^*)$ by tracing segments $t\in T^*$ in~$\arr(L)$
in $O(|T^*| + v_k) = O(m_k + m_{k+1} + v_k)$ time,
where $v_k$ denotes the number of vertices of~$\arr(L)$ we encounter.
Note that $\arr(T^*)$ consists of exactly $m_k + m_{k+1} + v_k$ vertices
and at most $m_k + m_{k+1} + 2v_k$ edges.

As $\arr(T^*)$ forms a plane graph, possibly being disconnected,
it turns out that its lower envelope~$\LE(T^*)$ can be obtained in time linear to its complexity.
Here, we make use of the following two algorithmic tools:
First,
a linear-time algorithm for computing the vertical decomposition (or the trapezoidation) 
of a simple polygon~\cite{c-tsplt-91,fm-tspep-84}
can be applied for computing the lower envelope of a connected plane graph.
\begin{lemma} \label{lem:envelope_plane_graph}
 Given a connected plane graph~$G$, consisting of $m$~line segments,
 its lower envelope~$\LE(G)$ can be computed in $O(m)$ time.
\end{lemma}
\begin{proof}
We extract the boundary~$B$ of the outer face of~$G$.
Considering the cyclic walk along~$B$,
we regard $B$ itself as a weakly simple polygon.
Find the leftmost and the rightmost vertices~$l,r$ of $B$ and 
take the lower chain~$C$ of~$B$ between $l$ and $r$.
Note that the chain $C$ may contain some edges that appear twice,
as $B$ is a weakly simple polygon.
Such doubled edges of $C$ can be resolved by an arbitrary perturbation avoiding overlaps.
So, we assume $C$ is a simple chain of line segments.
We then construct a simple polygon~$Q$ by attaching two sufficiently long vertical segments
below $l$ and $r$, respectively, and closing it with a third segment $b$ at the bottom.
Then, we compute the vertical decomposition of~$Q$ into trapezoids,
which adds vertical chords inside~$Q$ through every reflex corner of~$Q$.
This can be done in linear time by the linear-time algorithm
for triangulating a simple polygon~\cite{c-tsplt-91}
and a linear-time reduction of the vertical decomposition to the triangulation~\cite{fm-tspep-84}.
It is now obvious that the lower envelope~$\LE(G)$ of~$G$ 
can be constructed from those trapezoids that are incident to the bottom edge~$b$.
The total time spent above is bounded by the number of edges in~$G$.
\end{proof}

Second, Asano et al.~\cite{aaghi-vdp-86} presented a linear-time algorithm
for computing the lower envelope of disjoint line segments,
provided their endpoints are sorted.
It is not difficult to see that their algorithm also works
even if we replace ``segments'' by ``monotone chains.''
\begin{lemma}[Asano et al.~\cite{aaghi-vdp-86}] \label{lem:envelope_chains}
 Let $C_1, C_2, \ldots, C_l$ be mutually disjoint monotone chains with $m$ line segments in total.
 If a sorted list of their endpoints is given,
 then their lower envelope~$\LE(\bigcup_i C_i)$ can be computed in $O(m)$ time.
\end{lemma}

Back to our problem,
we apply Lemma~\ref{lem:envelope_plane_graph} to each connected component of~$\arr(T^*)$,
resulting in disjoint monotone chains $C_1, C_2, \ldots$, whose lower envelope is~$\LE(T^*)$.
Recall that we already have the sorted list of endpoints of~$T^*$,
since those endpoints have been obtained by walking along $H_k$ and $H_{k+1}$ in~$\arr(L \cup H)$
and applying the $\beta$-shifting map~$\tau_\beta$,
and the map~$\tau_\beta$ preserves the order along~$H_k$ and $H_{k+1}$.
Hence, we can extract a sorted list of endpoints of the chains~$C_i$
in additional $O(m_k + m_{k+1})$ time,
which allows us to apply Lemma~\ref{lem:envelope_chains} to obtain~$\LE(T^*)$.
The total time for this third step is proportional to the number of edges in~$\arr(T^*)$,
so $O(m_k + m_{k+1} + v_k)$ time.

Finally,
to compute~$\LE^+_k$, we linearly scan $\LE^+_{k+1}$ and $\LE(T^*)$, simultaneously.
This takes time linear to the total complexity of $\LE^+_{k+1}$ and $\LE(T^*)$.
Note that $\LE^+_{k+1}$ consists of 
$O(|T^+_{k+1}| \alpha(|T^+_{k+1}|)) = O(m_{k+1} \alpha(m_{k+1}))$ edges
since it is the lower envelope of line segments~\cite{sa-dsstga-95}.
So, the total time we spend to incrementally construct~$\LE^+_k$ from~$\LE^+_{k+1}$
for each~$1\leq k\leq n$
is bounded by $O(n + m_k + m_{k+1}\alpha(m_{k+1}) + v_k)$.

By iterating~$k$ from~$n$ down to~$1$, we conclude our decision algorithm.
\begin{theorem} \label{thm:afixed_decision}
 Given a set~$P$ of $n$~points, an angle~$\beta$, and a parameter~$\wparam$,
 we can decide whether or not $\wparam \geq w^*$
 in $O(n^2 \alpha(n))$ time and $O(n^2)$ space.
\end{theorem}
\begin{proof}
The correctness of the algorithm is discussed above.
We thus focus on the analysis of the running time.

The initialization is done in $O(n^2)$ time.
To compute $\LE^+_1, \LE^+_2, \ldots, \LE^+_n$,
we spend $O(n + m_k + m_{k+1}\alpha(m_{k+1}) + v_k)$ time for each $k=1, \ldots, n$
as analyzed above.
Obviously, we have
\[
 \sum_{k=1}^n ( n + m_k + m_{k+1}\alpha(m_{k+1}) + v_k) 
 = O(n^2 + \sum_{k=1}^n m_k \alpha(m_k) + \sum_{k=1}^n v_k).
\]
Now, observe that $\sum_k m_k = O(n^2)$ and $\sum_k v_k = O(n^2)$:
The former is obvious since $m_k$ counts the number of vertices in~$\arr(L\cup H)$ along~$H_k$.
To see the latter, observe that each vertex~$u$ in~$\arr(L)$ is counted exactly twice
in the sum~$\sum_k v_k$ by Lemma~\ref{lem:afixed_decision_incremental}.
Therefore, the total time to compute the envelopes~$\LE^+_k$, $\UE^+_k$, $\LE^-_k$, and $\UE^-_k$ 
for all $k=1,\ldots, n$ is $O(n^2 \alpha(n))$.
Note that the total complexity of all these envelopes is also bounded by $O(n^2 \alpha(n))$.

Our algorithm then merges $\LE^+_k$ and $\LE^-_k$ to obtain $\LE(T^+_k \cup T^-_k)$
and also $\UE^+_k$ and $\UE^-_k$ to obtain $\UE(T^+_k \cup T^-_k)$
for all~$k$.
This can be done in time linear to the complexity of the envelopes.
The final decision of whether or not $\wparam \geq w^*$
is then made by explicitly testing the width function determined by
each pair of extreme points of the second strip,
which are described by $\LE(T^+_k \cup T^-_k)$ and $\UE(T^+_k \cup T^-_k)$.
Hence, the total running time of the decision algorithm is bounded by $O(n^2 \alpha(n))$.

As our algorithm precomputes the envelopes and uses them later,
a straightforward implementation requires $O(n^2 \alpha(n))$ space
to store all the envelopes.
This can be improved to $O(n^2)$ as follows:
consider the envelopes~$\LE^+_1, \ldots, \LE^+_n$.
Recall that they form hierarchical layers in the plane, 
as described in Lemma~\ref{lem:afixed_decision_incremental}.
We store their union as a plane graph~$G$, whose edges are straight line segments.
Then, observe that each vertex of~$G$ is either
\begin{enumerate}[(1)] \denseitems
 \item a vertex from~$\arr(L)$, 
 \item an endpoint of a segment $t\in T^+_k$ for some~$k$, which is also a member of~$\Xi$, or
 \item a perpendicular foot of a vertex of the second type~(2).
\end{enumerate}
The degree of each vertex is at most four if it is of type~(1);
at most three, otherwise.
Hence, the number of vertices and edges in~$G$ is bounded by~$O(n^2)$.
Finally, each~$\LE^+_k$ can be traversed in~$G$ in time linear to the complexity of $\LE^+_k$
by labeling each edge~$e$ of~$G$ with an interval~$[k_0, k_1]$
such that $e$ belongs to $\LE^+_k$ for all $k_0 \leq k\leq k_1$.
All this additional construction can be done in $O(n^2\alpha(n))$ time,
while the space requirement is reduced to~$O(n^2)$.
\end{proof}

\subsection{First phase of the optimization algorithm}
From now on, we describe our optimization algorithm.
Its first phase is done as follows.

Let $W_2$ be the set of all pairwise distances among points in~$P$.
We first obtain two consecutive values $w'_0 < w'_1 \in W_2$ such that $w^* \in (w'_0, w'_1]$.
This is easily done in $O(n^2 \alpha(n) \log n)$ time 
by sorting~$W_2$ and performing a binary search on~$W_2$
using our decision algorithm presented in Theorem~\ref{thm:afixed_decision}.
Next, let $W_3$ be the set of $n\binom{n}{2}$ values obtained as follows:
for any pair~$p,q\in P$ with~$p\neq q$,
collect the distances from each $r\in P$ to the line through~$p$ and~$q$.
We then find two consecutive values $w''_0 < w''_1 \in W_3$ such that $w^* \in (w''_0, w''_1]$.
This can be done in $O(n^2 \alpha(n) \log n)$ time 
by the technique of Glozman et al.~\cite{gks-sgsopsm-98},
again using our decision algorithm.
Note that the two-strip of width~$w''_1$ is the best solution of case~(i) of Lemma~\ref{lem:afixed_conf}.
We then choose $w_0 := \max \{w'_0, w''_0\}$ and $w_1 := \min\{w'_1, w''_1\}$,
and obtain:
\begin{lemma} \label{lem:afixed_w0w1}
 In $O(n^2 \alpha(n) \log n)$ time, we can find two values $w_0 \leq w_1$ such that
 $w_0 < w^* \leq w_1$ and no member in $W_2\cup W_3$ lies in $(w_0, w_1)$.
\end{lemma}

\subsection{Second phase of the optimization algorithm}
For each $p\in P$,
let $w^*_p$ be the minimum possible width of constrained two-strips~$(\strip_1, \strip_2)$
such that $p$~lies on the boundary of~$\strip_1$.
It is obvious that $w^* = \min_{p\in P} w^*_p$.
The second phase of our algorithm computes the exact value of~$w^*_p$,
if $w^*_p < w_1$; or reports $w^*_p \geq w_1$, otherwise.
Note that, if $w^*_p < w_1$, then the corresponding optimal two-strip
falls in case~(ii) described in Lemma~\ref{lem:afixed_conf}
by Lemma~\ref{lem:afixed_w0w1}.
In the following, let $p\in P$ be fixed and called the \emph{pivot}.

\paragraph*{Updates in the sweeping process with fixed width.}

%

Before describing the algorithm,
we discuss essential ingredients of its correctness,
based on Lemma~\ref{lem:afixed_w0w1}.
Let $w \in (w_0, w_1)$ be any value.
We consider the sweeping process with fixed width~$w$ and fixed pivot~$p$,
as described at the beginning of Section~\ref{sec:afixed_decision}.
(See \figurename~\ref{fig:sweep}.)
Recall that the first strip~$\strip(\theta)$ is determined by
two directed lines~$\ell(\theta)$ and $\ell^+(\theta)$
such that $\ell(\theta)$ goes through~$p$
and $\ell^+(\theta)$ is at distance~$w$ to the left of~$\ell(\theta)$,
and the second strip~$\stripbar(\theta)$ in orientation~$\theta+\beta$
encloses the rest of points in~$P\setminus \strip(\theta)$.
Let $P(\theta) := P \cap \strip(\theta)$.
Then, during this sweeping process as $\theta$ increases, 
$P(\theta)$ undergoes a sequence of \emph{updates} (deletions and insertions).
We identify each update by a pair of its involved point~$r\in P$ and its \emph{type}
determined by one of the four combinations of the following:
\begin{itemize}\denseitems
 \item 
  An update is \emph{right} if it happens when $\ell(\theta)$ hits~$r$;
  or \emph{left} when $\ell^+(\theta)$ hits~$r$
 \item An update is \emph{leaving} if $r$ is being deleted from~$P(\theta)$;
  or \emph{approaching}, otherwise.
\end{itemize}
Thus, two updates are the same if their involved points and their types are equal.

Let $\Upd_w$ be the set of those updates occurred on~$P(\theta)$ 
during the sweeping process with fixed width~$w$ over $\theta \in [0,2\pi)$.
Observe that there are two possibilities for each~$r\in P\setminus\{p\}$:
By Lemma~\ref{lem:afixed_w0w1},
the distance between~$r$ and the pivot~$p$ is
either at most~$w_0$ or at least~$w_1$.
Thus, if $r$ falls in the former case, there are exactly two updates for~$r$ in~$\Upd_w$
whose types are right leaving and right approaching;
in the latter case, there are exactly four updates for~$r$ in~$\Upd_w$
with each of the four possible types.
This implies that the set~$\Upd_w$ is invariant under the choice of~$w \in (w_0, w_1)$,
so we write $\Upd = \Upd_w$ for any $w \in (w_0, w_1)$.

Fix an arbitrary right leaving update~$\upd_0 \in \Upd$ in which $r_0 \in P\setminus\{p\}$ is involved,
and assume that both~$p$ and~$r_0$ lie along~$\ell(0)$ in this order,
that is, $p$ and $r_0$ lie on the horizontal line~$\ell(0)$ and $r_0$ is to the right of~$p$;
this can be easily achieved by a proper rotation of the axes.
For~$\upd\in \Upd$ and~$w \in (w_0, w_1)$, 
let $\phi_\upd(w)\in[0,2\pi)$ be the direction 
at which $\upd$ occurs during the sweeping process with fixed width~$w$.
From the above discussion, we know that
$\phi_\upd$ is a well-defined function from~$(w_0, w_1)$ to~$[0,2\pi)$.
Lemma~\ref{lem:afixed_w0w1} implies the following.
\begin{lemma} \label{lem:afixed_phi}
 There is no $w \in (w_0, w_1)$ such that
 $\phi_{\upd}(w) = \phi_{\upd'}(w)$
 for any two distinct $\upd, \upd'\in \Upd$.
 Moreover, for each~$\upd\in \Upd$,
 $\phi_\upd(w)$ is either constant if $\upd$ is right,
 continuously increasing if $\upd$ is left leaving,
 or continuously decreasing if $\upd$ is left approaching.
\end{lemma}
\begin{proof}
First, note that $\phi_\upd(w)$ for any right update~$\upd\in \Upd$ is a constant function,
since those right updates are determined solely by the rotation of the line~$\ell(\theta)$
through the pivot~$p$,
independently of the choice of~$w\in (w_0, w_1)$.
This also proves the second part of the lemma for right updates.

Suppose $\phi_{\upd}(w) = \phi_{\upd'}(w)$
for some $w\in (w_0, w_1)$ and $\upd, \upd' \in \Upd$ with $\upd \neq \upd'$.
Let $\phi := \phi_{\upd}(w) = \phi_{\upd'}(w)$, and
let $r, r'\in P$ be the points involved in~$\upd$ and~$\upd'$, respectively.
By the above argument and the general position assumption, 
note that not both of~$\upd$ and~$\upd'$ are right updates,
so there are two cases: either one of the two is a right update or both are left updates.

In the former case, if $\upd$ is a right update,
both the pivot~$p$ and~$r$ lie on~$\ell(\phi)$,
and the distance from~$r'$ to line~$\ell(\phi)$ through~$p$ and~$r$ is exactly~$w$.
Thus, $w\in W_3$.
Since $w_0 < w < w_1$, 
this leads to a contradiction to Lemma~\ref{lem:afixed_w0w1}.
The other case where $\upd'$ is a right update can also be handled analogously.

Now, consider the latter case where both~$\upd$ and~$\upd'$ are left updates.
If $r \neq r'$, then the distance from~$p$ to the line through~$r$ and $r'$
is exactly~$w$, and we have $w\in W_3$,
leading to a contradiction to Lemma~\ref{lem:afixed_w0w1} as above.
Thus, we have $r = r'$, so $\upd$ and~$\upd'$ are leaving and approaching,
as $\upd \neq \upd'$.
This, however, implies that $\width_\theta(\{p, r\})$
attains its maximum~$w$ at~$\theta = \phi$.
As $\width_\theta(\{p, r\})$ is represented as a sinusoidal function of~$\theta$~\cite{Bae2020},
the value~$w$ is indeed the distance between $p$ and $r$, so $w\in W_2$.
This leads to a contradiction to Lemma~\ref{lem:afixed_w0w1},
as $w_0 < w < w_1$.

This completes the proof of the first part of the lemma.
By this, we have $\phi_\upd(w) \neq 0$ for any $\upd\in \Upd\setminus\{\upd_0\}$
and any $w\in (w_0, w_1)$.
For any left update~$\upd\in \Upd$ with its involving point~$r\in P$,
observe that 
the function~$\phi_\upd$ is indeed the inverse of 
$\width_\theta(\{p, r\})$ 
as a function restricted for those~$\theta$ such that the function value lies in~$(w_0, w_1)$,
that is,
 \[ \width_{\phi_\upd(w)}(\{p, r\}) = w\]
for any $w\in (w_0, w_1)$.
This implies the second part of the lemma for left updates.
\end{proof}

For $w\in (w_0, w_1)$,
we consider a total order~$\prec_w$ on~$\Upd$ such that
$\upd \prec_w \upd'$ if and only if
$\phi_\upd(w) < \phi_{\upd'}(w)$.
Note that its totality is guaranteed by Lemma~\ref{lem:afixed_phi}
and $\upd_0$ is the least element in~$\Upd$ under~$\prec_w$.
Lemma~\ref{lem:afixed_phi} further implies that
the ordering~$\prec_w$ on~$\Upd$ remains the same over all~$w\in (w_0, w_1)$:
assuming any swap between~$\prec_w$ and~$\prec_{w'}$ for $w_0 < w < w' < w_1$,
one can face with some $w'' \in (w, w')$ and $\upd,\upd'\in \Upd$ 
such that $\phi_{\upd}(w'') = \phi_{\upd'}(w'')$,
due to the continuity of functions~$\phi_\upd$ and~$\phi_\upd$,
so a contradiction.

Hence, we have a universal total ordering~$\prec$ on~$\Upd$ 
such that $\prec\; = \;\prec_w$ for any~$w \in (w_0, w_1)$.
Let $\upd_0, \upd_1, \ldots, \upd_{m-1} \in \Upd$ be
the updates in~$\Upd$ listed in this order~$\prec$,
where $m:= |\Upd|$.
For each $0\leq i\leq m-1$,
let $I_i:= \{\phi_{\upd_i}(w) \mid w_0 < w <w_1\}$.
Lemma~\ref{lem:afixed_phi} implies that
$I_i$ consists of a single element if $\upd_i$ is a right update;
otherwise, $I_i$ forms an open interval if $\upd_i$ is a left update.
The following summarizes more implications of Lemma~\ref{lem:afixed_phi}
about the intervals~$I_i$.
Two intervals~$I$ and $I'$ are said to be \emph{properly nested} 
if one includes the other, say $I' \subset I$, in such a way that
both endpoints of~$I'$ lie in the relative interior of~$I$.
\begin{lemma} \label{lem:afixed_update}
 Two intervals~$I_i$ and $I_j$ are never properly nested.
 If $I_i$ and $I_j$ overlap,
  then either both of~$\upd_i$ and $\upd_{j}$ are left leaving or
  both are left approaching.
\end{lemma}
\begin{proof}
Suppose that $I_i$ and $I_j$ are properly nested and 
$I_i \subset I_j$.
Note that not both~$\upd_i$ and $\upd_j$ are right updates,
since, otherwise, both~$I_i$ and~$I_j$ are singletons 
and cannot be properly nested.
Consider the graphs~$\gamma_i$ of~$\{\theta = \phi_{\upd_i}(w)\}$ and
$\gamma_j$ of~$\{\theta = \phi_{\upd_j}(w)\}$ drawn in the $(w, \theta)$-plane over $w\in (w_0, w_1)$.
Then, we observe that $\gamma_i$ and $\gamma_j$ must cross,
regardless of the types of~$\upd_i$ and~$\upd_j$,
since $\gamma_i$ and $\gamma_j$ are continuous by Lemma~\ref{lem:afixed_phi}.
See \figurename~\ref{fig:afixed_update}(a)
for an illustration to the case where both~$\upd_i$ and~$\upd_j$ are left leaving updates.
This implies the existence of $w' \in (w_0, w_1)$ such that
$\phi_{\upd_i}(w') = \phi_{\upd_j}(w')$,
a contradiction to Lemma~\ref{lem:afixed_phi}.
Hence, $I(\upd)$ and $I(\upd')$ are not properly nested.
This proves the first statement of the lemma

\begin{figure}[htb]
  \centerline{\includegraphics[width=.90\textwidth]{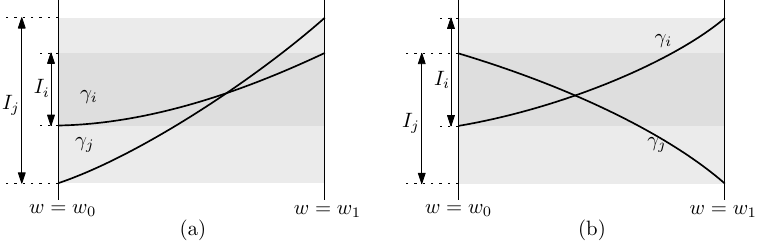}}
  \caption{Proof of Lemma~\ref{lem:afixed_update}.}
  \label{fig:afixed_update}
\end{figure}

Next, assume that $I_i\cap I_j \neq \emptyset$.
This rules out the possibility that one of~$\upd_i$ and $\upd_j$ is a right update
by the first statement of the lemma and also by the general position assumption.
So, both are left updates.
Suppose $\upd_i$ is leaving and $\upd_j$ is approaching.
We then observe that $\gamma_i$ and $\gamma_j$ must cross,
since $\phi_{\upd_i}$ is continuously increasing and $\phi_{\upd_j}$ is continuously decreasing
by Lemma~\ref{lem:afixed_phi}.
See \figurename~\ref{fig:afixed_update}(b).
This again leads to a contradiction to Lemma~\ref{lem:afixed_phi}.
Therefore, either both of~$\upd_i$ and~$\upd_j$ are left leaving updates 
or both are left approaching updates.
\end{proof}

For $0\leq i\leq m-1$,
let $P_i$ be the resulting set after executing
the first $i+1$ updates~$\upd_0, \ldots, \upd_i$
on the subset of points in~$P$ lying on or to the left of~$\ell(0)$
whose distances to~$\ell(0)$ are at most~$w_0$.
Note that $P_0 = (\strip_0 \cap P) \setminus\{r_0\}$
where $\strip_0$ denotes the horizontal strip of width~$w_0$
such that $\ell(0)$ bounds~$\strip_0$ from below.
Let $Q_i := P\setminus P_i$, and define
 \[ \wf_i(\theta) := \width_\theta(P_i) \quad \text{and} \quad 
    \wfg_i(\theta) := \width_{\theta+\beta}(Q_i),\]
for $\theta \in [0,2\pi)$.
%
Let $r_i \in P\setminus \{p\}$ be the point involved in~$\upd_i$.
\begin{lemma} \label{lem:afixed_width_leaving}
 For any left leaving update~$\upd_i \in \Upd$,
 $\wf_{i-1}(\theta) = \width_\theta(\{p,r_i\})$ over~$\theta\in I_i$,
 and is an increasing function over~$I_i$ whose infimum and supremum are $w_0$ and $w_1$, respectively.
\end{lemma}
\begin{proof}
Note that $P_i = P_{i-1} \setminus \{r_i\}$.
Let $\delta_i \colon I_i \to (w_0, w_1)$ be a function such that
 \[ \delta_i(\theta) = \width_\theta(\{p, r_i\}).\]
As also discussed in the proof of Lemma~\ref{lem:afixed_phi},
$\delta_i$ is the inverse function of~$\phi_{\upd_i}$.
By Lemma~\ref{lem:afixed_phi},
we know that $\phi_{\upd_i} \colon (w_0, w_1) \to I_i$ is a continuous increasing function
and thus a bijection.
Hence, $\delta_i$ is also increasing over~$I_i$ whose infimum and supremum are
$w_0$ and $w_1$, respectively.

In the rest of the proof, we prove that
 \[ \wf_{i-1}(\theta) = \delta_i(\theta) = \width_\theta(\{p, r_i\})\]
for $\theta \in I_i$.
Consider the sweeping process
with fixed width~$w$ and fixed pivot~$p$, as described above.
As discussed above, Lemma~\ref{lem:afixed_phi} implies that
 \[ \phi_{\upd_0}(w) < \phi_{\upd_1}(w) < \cdots < \phi_{\upd_{m-1}}(w) \]
for any~$w\in (w_0, w_1)$.
Hence, at~$\theta = \phi_{\upd_i}(w)$,
the first $i$~updates~$\upd_0, \upd_1 \ldots, \upd_{i-1}$ have been placed,
so we have
 \[ \width_{\phi_{\upd_i}(w)}(P_{i-1}) = \width_{\phi_{\upd_i}(w)}(\{p, r_i\}) 
  = \delta_i(\phi_{\upd_i}(w)).\]
Furthermore,
by definition of~$I_i$,
$I_i$ is the image of~$(w_0, w_1)$ under~$\phi_{\upd_i}$.
This implies that 
$\wf_{i-1}(\theta) = \width_{\theta}(P_{i-1}) = \delta_i(\theta)$
for any $\theta \in I_i$, and the lemma follows.
\end{proof}

\paragraph*{Description of algorithm.}
The second phase of our algorithm simulates a similar sweeping process as before,
but with the first strip~$\strip(\theta)$ having \emph{variable} width:
Let $\wf \colon [0,2\pi) \to \Real$ be a function,
which will be specified later.
We redefine $\ell^+(\theta)$ to be the line at distance~$\wf(\theta)$
to the left of~$\ell(\theta)$,
and thus $\strip(\theta)$ to have width~$\wf(\theta)$.
The second strip~$\stripbar(\theta)$ in orientation~$\theta + \beta$ is determined as before
to tightly enclose the rest of points in~$P \setminus \strip(\theta)$.
Let $\wfg(\theta) := \width(\stripbar(\theta))$.
This way, the process is completely determined by the width function~$\wf(\theta)$.

Our width function~$\wf(\theta)$ will be fully determined
by when to execute each update~$\upd_i\in \Upd$.
For $0\leq i \leq m-1$,
let $\phi_i$ be the direction
at which the $i$-th update~$\upd_i \in \Upd$ is executed in our algorithm.
We choose the $\phi_i$'s by the following rules:
\begin{itemize}\denseitems
 \item If $\upd_i$ is a right update,
  $\phi_i$ is the only direction in~$I_i$, that is, $I_i = \{\phi_i\}$.
 \item If $\upd_i$ is a left approaching update,
  $\phi_i$ is chosen to be the larger endpoint of~$I_i$.
 \item If $\upd_i$ is a left leaving update,
  $\phi_i$ is chosen to be the smallest direction~$\theta$ 
  such that $\wf_{i-1}(\theta) = \wfg_{i-1}(\theta)$ over~$\theta \in I_i$, if exists;
  otherwise, $\phi_i$ is the larger endpoint of~$I_i$.
\end{itemize}
Note that $\phi_0 = 0$ and let $\phi_m := 2\pi$.
It is obvious that either $\phi_i \in I_i$ or
$\phi_i$ is the larger endpoint of~$I_i$.
Less obvious is that the resulting $\phi_i$'s indeed obey the ordering~$\prec$ of~$\Upd$.
\begin{lemma} \label{lem:afixed_order_phi}
 It holds that 
  $0 = \phi_0 <\phi_1\leq \phi_2 \leq \cdots \leq \phi_{m-1} \leq \phi_m = 2\pi$.
\end{lemma}
\begin{proof}
Pick any $0\leq i \leq m-2$.
If $I_i \cap I_{i+1}=\emptyset$, then we immediately have $\phi_{i} \leq \phi_{i+1}$,
since each $\phi_i$ is chosen to be either $\phi_i\in I_i$ or
the larger endpoint of~$I_i$ according to the above rules.
So, suppose $I_i$ and $I_{i+1}$ have a nonempty intersection.
Then Lemma~\ref{lem:afixed_update} implies that
both $\upd_i$ and $\upd_{i+1}$ are left leaving updates
or both are left approaching updates.
In the latter case, $\phi_i$ and $\phi_{i+1}$ are chosen to be
the larger endpoints of $I_i$ and $I_{i+1}$ by the rules, and thus $\phi_i \leq \phi_{i+1}$.

What remains is the former case where both $\upd_i$ and $\upd_{i+1}$ are left leaving updates.
Let $J := I_i \cap I_{i+1}$.
If either $\phi_i \notin I_{i+1}$ or $\phi_{i+1} \notin I_i$, then we are obviously done,
so assume that $\phi_i \in I_{i+1}$ and $\phi_{i+1} \in I_i$.
Suppose to the contrary that $\phi_i > \phi_{i+1}$.
From the rules, 
this implies that $\phi_{i+1} \in J$ is the smallest direction over range~$I_{i+1}$ such that
$\wf_{i}(\phi_{i+1}) = \wfg_{i}(\phi_{i+1})$.
We claim that
\begin{quote}
 (*) $\wf_{i-1}(\theta) > \wf_{i}(\theta)$ and $\wfg_{i-1}(\theta) \leq \wfg_{i}(\theta)$
 for any~$\theta \in J$.
\end{quote}
We first show a contradiction, assuming the above claim (*) is true,
and next we prove the claim.

On one hand, the claim~(*) implies 
 \[ \wfg_{i-1}(\phi_{i+1}) \leq \wfg_{i}(\phi_{i+1}) = \wf_{i}(\phi_{i+1}) < \wf_{i-1}(\phi_{i+1}).\]
On the other hand, letting $\theta_0$ be the smaller endpoint of~$I_i$,
observe that 
 \[  \wf_{i-1}(\theta_0) = w_0 < \wfg_{i-1}(\theta_0),\]
since $\theta_0 = \phi_{\upd_{i-1}}(w_0)$
and $w_0 < w^*$.
As both functions~$\wf_{i-1}$ and~$\wfg_{i-1}$ are continuous over~$J$,
this implies the existence of~$\phi'$ with $\theta_0 < \phi' < \phi_{i+1}$ such that
$\wf_{i-1}(\phi') = \wfg_{i-1}(\phi')$.
This leads to a contradiction to the definition of~$\phi_i$, since $\phi' < \phi_{i+1} < \phi_i$.
So, we have $\phi_i \leq \phi_{i+1}$, and the lemma follows.

Now, we prove the claim~(*).
Since both $\upd_i$ and $\upd_{i+1}$ are leaving updates,
we have 
 \[ P_i = P_{i-1} \setminus \{r_i\}, \quad P_{i+1} = P_i \setminus \{r_{i+1}\}, \quad
  Q_i = Q_{i-1} \cup \{r_i\}, \quad \text{ and } \quad Q_{i+1} = Q_i \cup \{r_{i+1}\}.\]
From the definitions of~$\wf_{i-1}$ and~$\wfg_{i-1}$, this already implies that
$\wf_{i-1}(\theta) \geq \wf_i(\theta)$ and $\wfg_{i-1}(\theta) \leq \wfg_i(\theta)$
for $\theta \in J$.
The strict inequality $\wf_{i-1}(\theta) > \wf_{i}(\theta)$
is observed from the fact that 
 \[ \wf_{i-1}(\theta) = \width_\theta(\{p, r_i\}) \quad \text{and} \quad
    \wf_{i}(\theta) = \width_\theta(\{p, r_{i+1}\}) \]
by Lemma~\ref{lem:afixed_width_leaving},
and that $r_i \neq r_{i+1}$ since both $I_i$ and $I_{i+1}$ are left leaving.
\end{proof}
The function~$\wf(\theta)$ is then set up as follows:
$\wf(0) := w_0$ and
$\wf(\theta) := \max\{ w_0, \wf_i(\theta)\}$ for $\theta \in (\phi_i, \phi_{i+1}]$
and~$0\leq i\leq m-1$.
We then obtain a conditional correctness of our algorithm.
\begin{lemma} \label{lem:afixed_correctness}
 Suppose $w^*_p < w_1$, and 
 let $\theta^*$ and $\theta^*+\beta$ be the directions of
 the bounding lines of a corresponding two-strip of width~$w^*_p$
 such that the pivot~$p$ lies on the right bounding line of direction~$\theta^*$.
 If $\theta^*\notin I_i$ for all left approaching updates~$\upd_i\in\Upd$,
 then there is a left leaving update~$\upd_j \in \Upd$ such that
 $\theta^* = \phi_j \in I_j$ and
 $w^*_p = \wf(\theta^*) = \wf_{j-1}(\theta^*) = \wfg_{j-1}(\theta^*) = \wfg(\theta^*)$.
\end{lemma}
\begin{proof}
Suppose $w^*_p < w_1$ 
and $\theta^*$ avoids all the intervals of left approaching updates in~$\Upd$. 
Consider the sweeping process with fixed width~$w^*_p$ for the unknown~$w^*_p$.
At $\theta = \theta^*$ in the process,
the two strips in orientation~$\theta^*$ and $\theta^*+\beta$, respectively, have equal width~$w^*_p$ and
each of their four bounding lines contains a point in~$P$
by Lemma~\ref{lem:afixed_conf}(ii).
This implies that $\theta^* = \phi_{\upd_j}(w^*_p) \in I_j$
for some left leaving update~$\upd_j \in \Upd$.
Hence, the first strip in orientation~$\theta^*$ is determined by
the pivot~$p$ and the point~$r_j$ involved in~$\upd_j$.
Furthermore, the first strip encloses~$P_{j-1} = P_j \cup \{r_j\}$,
which implies that the second strip encloses~$P\setminus P_{j-1} = Q_{j-1}$.
Therefore, we have
 \[ \width_{\theta^*}(\{p, r_j\}) = \width_{\theta^*}(P_{j-1}) 
   = \width_{\theta^*+\beta}(Q_{j-1}) = w^*_p.\]

Since $\upd_j$ is a left leaving update,
Lemma~\ref{lem:afixed_width_leaving} implies that
 \[ \wf_{j-1}(\theta^*) = \width_{\theta^*}(\{p, r_j\}),\]
while we have
 \[ \wfg_{j-1}(\theta^*) = \width_{\theta^*+\beta}(Q_{j-1})\]
by definition.
Hence, the above discussion concludes that
 \[ \wf_{j-1}(\theta^*) = \wfg_{j-1}(\theta^*) = w^*_p.\]
 
By our rules, $\phi_j$ is the smallest $\theta \in I_j$ such that
$\wf_{j-1}(\theta) = \wfg_{j-1}(\theta)$, if exists; or the larger endpoint of~$I_j$, otherwise.
From the existence of~$\theta^* \in I_j$, we exclude the possibility of the latter case,
so $\phi_j \in I_j$.
Observe that $\phi_j$ indeed minimizes the value of $\wf_{j-1}(\theta)$
over all those solutions to the equation $\wf_{j-1}(\theta) = \wfg_{j-1}(\theta)$
as $\wf_{j-1}$ is an increasing function by Lemma~\ref{lem:afixed_width_leaving}.
This suffices to see that $\theta^* = \phi_j$
and $w^*_p = \wf(\theta^*) = \wf_{j-1}(\theta^*) = \wfg_{j-1}(\theta^*) = \wfg(\theta^*)$
by Lemma~\ref{lem:afixed_order_phi} and our definition of the width function~$\wf(\theta)$. 
This completes the proof of the lemma.
\end{proof}
Thus, we can compute~$w^*_p$ and its corresponding two-strip
by checking each~$\phi_i$ such that $\phi_i \in I_i$ and $\upd_i\in\Upd$ is a left leaving update,
provided the condition of Lemma~\ref{lem:afixed_correctness} is satisfied.
The other case, where $w^*_p$ is \emph{not} determined by left leaving updates,
can be handled by a \emph{reversed} sweeping process that rotates $\strip(\theta)$
clockwise by decreasing~$\theta$ from $2\pi$ to $0$;
note that in this reversed process each approaching update becomes a leaving update, and vice versa.

Now, the detailed implementation is presented.
Simulating the sweeping process with function~$\wf(\theta)$ can be done by maintaining
a dynamic set~$Q$, representing $P\setminus \strip(\theta)$, and its convex hull~$\conv(Q)$.
First, we compute the updates~$\upd_0, \upd_1, \ldots, \upd_{m-1} \in \Upd$
together with their intervals~$I_i$, and also precompute~$\phi_i$ 
for all right updates and left approaching updates~$\upd_i \in \Upd$.
Initially, $Q = Q_0$ and $Q = Q_i$ while we are in~$\theta \in (\phi_i, \phi_{i+1}]$
for each $0\leq i \leq m$.
We also maintain the two extreme points of~$Q$ that determine~$\stripbar(\theta)$:
this can be done by two types of queries on~$\conv(Q)$,
finding two tangents of~$\conv(Q)$ in a given direction and
finding the next extreme point of~$\conv(Q)$ neighboring the current one.
Each of these convex hull queries can be answered in $O(\log n)$ amortized time~\cite{bj-dpch-02}.

While we rotate $\strip(\theta)$ as increasing~$\theta$,
we execute updates~$\upd_i \in \Upd$ at $\theta = \phi_i$ if $\phi_i$ has already been computed.
Recall that only the execution times~$\phi_i$ of left leaving update~$\upd_i$
are not precomputed, so they are evaluated during the sweeping process:
Suppose the current direction~$\theta$ lies in~$I_i$ for a left leaving update~$\upd_i$
and the first $j\leq i$ updates $\upd_0, \ldots, \upd_{j-1}$ have already been executed,
that is, $Q = Q_{j-1}$ currently at~$\theta$ and $\wfg(\theta) = \wfg_{j-1}(\theta)$.
At this moment~$\theta$, note that $\theta \in I_j \cap I_i$
and $\upd_j$ is also a left leaving update by Lemma~\ref{lem:afixed_update}.
Hence, Lemma~\ref{lem:afixed_width_leaving} implies that
$\wf(\theta) = \wf_{j-1}(\theta) = \width_\theta(\{p, r_j\})$.
We then solve the equation $\wf_{j-1}(\varphi) = \wfg_{j-1}(\varphi)$.
Since the two functions $\wf_{j-1}$ and $\wfg_{j-1}$ are sinusoidal
over a range in which the two extreme points of~$Q_{j-1}$ do not change~\cite{Bae2020},
this can be done in time proportional to the number of such changes while $Q = Q_{j-1}$.
As soon as we find a solution~$\phi\in I_j$ such that $\wf_{j-1}(\phi) = \wfg_{j-1}(\phi)$,
we know that $\phi_j = \phi$ by our rules;
otherwise, $\phi_j$ is chosen to be the larger endpoint of~$I_j$.

Since $m = O(n)$, the overall time we spend is bounded by
$O(n \log n + E \log n)$,
where $E$ denotes the number of changes of the extreme points of~$Q$
that define the second strip~$\stripbar(\theta)$.
In the dual setting, as done for the decision algorithm,
those changes correspond to the vertices 
of the lower and upper envelopes of $O(n)$ line segments,
so we have $E = O(n \alpha(n))$~\cite{sa-dsstga-95}.
By iterating pivots $p\in P$,
the second phase of the algorithm can be implemented in $O(n^2 \alpha(n) \log n)$ total time.

Therefore, we conclude the following result.
\begin{theorem}\label{thm:afixed}
 Given a set~$P$ of $n$~points and a parameter~$\beta\in [0,\pi/2]$,
 the two-line center problem with a constraint that
 the resulting two lines should make an angle of~$\beta$
 can be solved in $O(n^2 \alpha(n)\log n)$ time
 using $O(n^2)$ space.
\end{theorem}

\bibliography{paper}
\bibliographystyle{abbrv}

\end{document}